\documentclass[11pt,a4paper,oneside]{article}
\usepackage[top=3cm, bottom=3cm, left=2.5cm, right=2.5cm]{geometry}
\usepackage[english]{babel}
\usepackage[utf8x]{inputenc}
\usepackage[T1]{fontenc}

\usepackage{amsthm,amsmath,amssymb,mathrsfs,dsfont,mathtools}
\usepackage{bm}
\usepackage{graphicx}
\usepackage[colorlinks=true, allcolors=blue]{hyperref}

\usepackage{booktabs,subcaption}
\usepackage{hyperref}
\usepackage[all]{nowidow}
\usepackage{setspace}

\usepackage[ruled,vlined]{algorithm2e}
\usepackage{tikz} 
\usepackage[sf]{titlesec}
\usepackage{sectsty}
\allsectionsfont{\centering \normalfont\scshape} 
\usepackage{lipsum}
\usepackage{adjustbox}
\usepackage{lineno} 
\usepackage{macros}
\usepackage{verbatim,amssymb,epsfig}  
\usepackage{amsmath, amsfonts, amstext,amsthm}
\usepackage{mathtools}
\usepackage{enumerate}
\usepackage{enumitem}
\usepackage{natbib}
\usepackage{graphicx}
\usepackage[title]{appendix}
\usepackage{subcaption}
\usepackage{booktabs} 

\usepackage{hyperref}
\usepackage{dsfont}
\hypersetup{
colorlinks = true,
allcolors = blue
}
\usepackage[font=small]{caption}
\usepackage{url} 
\usepackage{authblk}

\usepackage{multirow, makecell}

\usepackage{authblk}
\title{Logistic-Beta Processes for Dependent Random Probabilities with Beta
Marginals}
\date{}
\author[1]{Changwoo J. Lee \thanks{changwoo.lee@duke.edu}}
\author[2]{Alessandro Zito}
\author[3]{Huiyan Sang}
\author[1]{David B. Dunson} 
\affil[1]{Department of Statistical Science, Duke University}
\affil[2]{Department of Biostatistics, Harvard University}
\affil[3]{Department of Statistics, Texas A\&M University}

\theoremstyle{plain}

\newtheorem{theorem}{Theorem}[section]

\newtheorem{proposition}{Proposition}[section]
\newtheorem{corollary}{Corollary}[section]
\theoremstyle{definition}
\newtheorem{definition}[theorem]{Definition}

\theoremstyle{remark}

\def\T{{ \mathrm{\scriptscriptstyle T} }}

\def\lb{\normalfont{\textsc{lb}}} 
\def\lbp{\normalfont{\textsc{lbp}}}
\def\pipo{\normalfont{\pi_\textsc{po}}}
\def\pipg{\normalfont{\pi_\textsc{pg}}}
\def\pilb{\normalfont{\pi_\textsc{lb}}}
\def\pilb{\normalfont{\pi_\textsc{lb}}}

\begin{document}

\maketitle

\begin{abstract}
The beta distribution serves as a canonical tool for modeling probabilities in statistics and machine learning. However, there is limited work on flexible and computationally convenient stochastic process extensions for modeling dependent random probabilities.
We propose a novel stochastic process called the logistic-beta process, whose logistic transformation yields a stochastic process with common beta marginals.
Logistic-beta processes can model dependence on both discrete and continuous domains, such as space or time, and have a flexible dependence structure through correlation kernels. 
Moreover, its normal variance-mean mixture representation leads to effective posterior inference algorithms. 
We show how the proposed logistic-beta process can be used to design computationally tractable dependent Bayesian nonparametric models, including dependent Dirichlet processes and extensions.
We illustrate the benefits through nonparametric binary regression and conditional density estimation examples, both in simulation studies and in a pregnancy outcome application. 
\end{abstract}

\noindent%
\vfill
\section{Introduction}
\label{sec:1}
The beta distribution is widely used for modeling random probabilities owing to its flexibility, natural interpretation of parameters, and conjugacy in Bayesian inference. 
Beta distributions play a pivotal role in Bayesian nonparametric modeling \citep{Hjort2010-mg,Muller2015-rt,Ghosal2017-fk}. 
For instance, a broad class of random probability measures have stick-breaking representations in terms of beta random variables, including the popular Dirichlet processes and Pitman-Yor processes \citep{Ferguson1973-av, Sethuraman1994-ab,Perman1992-yh, Pitman1997-ym,Ishwaran2001-rd}. 
Other notable models that build upon the beta distribution include P\'olya trees \citep{Ferguson1974-tl} for nonparametric density estimation and beta processes \citep{Hjort1990-zd}, which were originally designed as priors for the cumulative hazard rate 
but have been used broadly in machine learning due to their link with latent feature models, such as the
Indian buffet process \citep{Griffiths2011-le}.

Stochastic processes with beta marginals play a pivotal role in extending such models to handle non-exchangeable and heterogeneous data sources. In the Bayesian nonparametric literature, a common solution lies in defining probability measure-valued stochastic processes whose realizations depend on covariates while maintaining marginal properties; examples include 
dependent Dirichlet processes  
\citep[DDP;][]{MacEachern1999-ui,Quintana2022-mc}, dependent P\'olya trees \citep{Trippa2011-sz} and dependent Indian buffet processes \citep{Perrone2017-ql}. 
Maintaining marginal properties is crucial since it not only simplifies how covariates influence the induced prior process but also clarifies the role of hyperparameters; see also \citet{MacEachern2000-lw}'s four desiderata for dependent Bayesian nonparametric models.

This article is motivated by the need to define flexible and computationally convenient classes of beta processes to substantially improve current approaches, such as copula-based 
formulations
\citep{MacEachern2000-lw,Arbel2016-or,De_Iorio2023-wc}. Copulas are highly flexible but lead to
challenging posterior computation, particularly in complex settings. 
\citet{Trippa2011-sz} and \citet{Bassetti2014-fa} instead 
extend bivariate beta distributions \citep{Olkin2003-du,Nadarajah2005-mx}, 
but this can only induce positive correlation. Other strategies rely on
latent binomial random variables \citep{Pitt2002-ue,Pitt2005-ww, Nieto-Barajas2012-fr} and the transformation of Gaussian autoregressive processes \citep{DeYoreo2018-zs}, but focus narrowly on discrete-time dependence.
Finally, constructions based on covariate-dependent ordering \citep{Griffin2006-kp}, local sharing of random components \citep{Chung2011-xq}, and Wright-Fisher diffusion \citep{Perrone2017-ql} are complex and
limited to specialized cases.

We propose a novel logistic-beta process, whose logistic transformation is a stochastic process with common beta marginals. 
The logistic-beta process builds on the normal variance-mean mixture representation of a univariate logistic-beta distribution \citep{Barndorff-Nielsen1982-xz}, and extends its hierarchical formulation to a stochastic process through correlation kernels in a manner related to a Gaussian process. 
As such, the proposed process can be flexibly defined in either discrete or continuous domains, and the implied dependence ranges from a perfect correlation to possibly negative correlations. Adopting the logistic-beta process prior for logit-transformed latent random probabilities provides a conditionally conjugate posterior sampling scheme that can directly exploit P\'olya-Gamma data augmentation strategies \citep{Polson2013-gb}. 

We provide two modeling examples utilizing the logistic-beta process: 
nonparametric regression for dependent binary data and conditional density estimation via a dependent Dirichlet process mixture with covariate-dependent weights and atoms.
There are few methods for covariate-dependent weights in the dependent Dirichlet process due to computational challenges \citep{Quintana2022-mc,Wade2025-ef}.
We illustrate how the proposed formulation leads to efficient posterior inference algorithms.
Although there are alternative Bayesian nonparametric models for conditional densities
\citep{Dunson2008-ds, Ren2011-el,Rigon2021-ir}, maintaining a marginal Dirichlet process structure has advantages in terms of understanding the roles of parameters involved in the model, ensuring that marginal prior properties are not affected by covariates in the absence of such information, and alleviation of the hyperparameter sensitivity problem highlighted by \citet{Wade2025-ef}.

\section{Logistic-beta processes}
\label{sec:2}
\subsection{Univariate logistic-beta and P\'olya mixing distributions}
We begin by reviewing the univariate logistic-beta distribution, which is also referred to as the type IV generalized logistic distribution \citep{Johnson1995-sk} or Fisher's z distribution \citep{Barndorff-Nielsen1982-xz} in the literature up to location-scale transformations. We say $\eta$ follows a univariate logistic-beta distribution with shape parameters $a,b>0$, if it has density 
\begin{equation}\label{eq:LB1}
    \pilb(\eta; a,b) = \frac{1}{B(a, b)}\Big(\frac{1}{1+e^{-\eta}}\Big)^{a} \Big(\frac{e^{-\eta}}{1+e^{-\eta}}\Big)^b, \quad \eta\in\mathds{R},
\end{equation}
where $B(a,b)$ is a beta function. 
When $a=b=1$, equation~\eqref{eq:LB1} reduces to the standard logistic distribution. 
Applying a logistic transformation $\sigma(x) = 1/(1+e^{-x})$ to a logistic-beta random variable yields a beta distribution $\sigma(\eta)\sim \mathrm{Beta}(a,b)$, hence the name logistic-beta.
We refer to \citet[][\S 23.10]{Johnson1995-sk} for further distributional details.

An important property of the logistic-beta distribution is its representation as a normal variance-mean mixture, as first shown by \citet{Barndorff-Nielsen1982-xz}. In particular, the density in equation~\eqref{eq:LB1} can be equivalently written as 
\begin{equation}
    \pilb(\eta; a,b) =\int_0^\infty N_1\left\{\eta; 0.5 \lambda(a-b), \lambda\right\}\pipo(\lambda ; a,b) \mathrm{d}\lambda,
   \label{eq:lbnormalvarmeanmixture}
\end{equation}
where $N_d(\bm\mu, \bm\Sigma)$ and $N_d(\bm\eta; \bm\mu, \bm\Sigma)$ denote the distribution and density of a $d$-dimensional multivariate normal with mean $\bm\mu$ and covariance $\bm\Sigma$, and $\pi_{\textsc{po}}(\lambda; a,b)$ denotes the density of a P\'olya distribution with shape parameters $a$ and $b$. This mixture representation is occasionally employed in regularized logistic regression problems \citep{Gramacy2012-dz,Polson2013-qa}. Figure~\ref{fig:lb_polya} shows density plots of the logistic-beta and corresponding P\'olya mixing distributions. 

\begin{figure}[t]
\includegraphics[width=\textwidth]{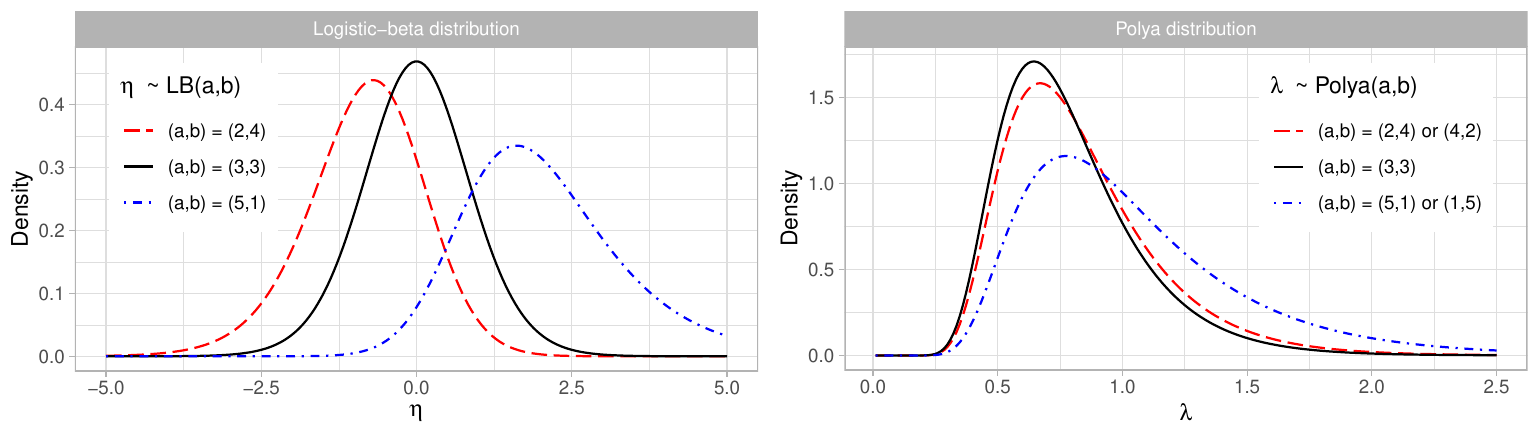}
    \caption{Density plot of univariate logistic-beta and associated P\'olya mixing distributions.}
    \label{fig:lb_polya}
\end{figure}

The P\'olya mixing distribution, denoted as $\lambda\sim \mathrm{Polya}(a,b)$, is defined as an infinite sum of exponential distributions $\sum_{k=0}^\infty 2\epsilon_k/\{(k+a)(k+b)\}$, with $\epsilon_k$ following standard exponential distributions independently for $k=0,1,\dots$. 
A truncated version of this expression can be used for approximate random variable generation.
The density function of a P\'olya can be written as an alternating series \citep{Barndorff-Nielsen1982-xz} \begin{equation}
    \pipo(\lambda ; a, b) = \sum_{k=0}^\infty (-1)^k \binom{a+b + k-1}{k}\frac{k+(a+b)/2}{B(a,b)}\exp\left\{-\frac{(k+a)(k+b)}{2}\lambda\right\}.
    \label{eq:polyadensity}
\end{equation} 
In practice, this representation brings numerical issues in evaluating the density, especially near the origin where the exponential term in equation~\eqref{eq:polyadensity} has a slow decay. However,  we derive a simple identity between P\'olya densities among parameter pairs $(a,b)$ that have fixed sum, which we use to develop posterior inference algorithms that avoid direct evaluation of equation~\eqref{eq:polyadensity}.

\begin{proposition}
\label{prop:polya_dens_identity}
The P\'olya density in equation~\eqref{eq:polyadensity} satisfies the following identity:
\begin{equation}
    \pipo(\lambda; a',b') = B(a,b)B(a',b')^{-1}\exp\{ \lambda(ab-a'b')/2\}\pipo(\lambda; a,b),
\end{equation}
where $(a',b')$ is any pair that satisfies $a'+b' = a+b$ and $a',b'>0$. 
\end{proposition}
See Section S.1 of the Appendix for all proofs, and additional details on the P\'olya distribution. 

\subsection{Multivariate logistic-beta distributions}

Based on the normal variance-mean mixture representation in equation~\eqref{eq:lbnormalvarmeanmixture}, we introduce a multivariate logistic-beta distribution that has the same univariate logistic-beta marginals.

\begin{definition} Let $\bfR$ be an $n\times n$ positive semidefinite correlation matrix. 
We say $\bm\eta$ follows an $n$-dimensional multivariate logistic-beta distribution with shape parameters $a,b>0$ and correlation parameter $\bfR$, denoted as $\bm\eta\sim \lb(a,b,\bfR)$, if 
\begin{equation}
 \bm\eta \mid \lambda \sim N_n\left\{0.5\lambda (a-b)\bm{1}_n,  \lambda \bfR\right\}, \quad \lambda \sim \mathrm{Polya}(a,b),
 \label{eq:mlbnormalvariancemeanmixture}
\end{equation} 
where $\bm{1}_n = (1,\dots,1)^\T$.
\end{definition}

Elementwise logistic transformation of a multivariate logistic-beta distribution induces a multivariate distribution on the unit hypercube $(0,1)^n$ with $\mathrm{Beta}(a,b)$ marginals. 
To see this, since $\bfR$ has a unit diagonal, each component $\eta_i$
is marginally logistic-beta distributed with shape parameters $a$ and $b$ for $i=1,\dots,n$.
Thus, logistic transformation $\{\sigma(\eta_1), \dots,\sigma(\eta_n)\}^\T$ follows a multivariate beta distribution with $\bfR$ controlling dependence.
See Figure~\ref{fig:multivariatelb} for illustrations. Note that $\lambda\bfR$ in equation~\eqref{eq:mlbnormalvariancemeanmixture} is a valid covariance matrix since $\lambda >0$ and $\bfR$ is a positive semi-definite correlation matrix. Hence, the eigenvalues of $\lambda\bfR$ are simply the eigenvalues $\bfR$ scaled by $\lambda$. Such a multivariate generalization is mentioned in \cite{Barndorff-Nielsen1982-xz}, and the density is studied in \cite{Grigelionis2008-um}.
However, there is no literature using this construction for multivariate beta modeling. Alternative multivariate constructions related to logistic-beta are proposed by \cite{Bradley2019-mx} for conjugate spatio-temporal modeling of multinomial data and by \cite{Kowal2019-od} for temporal dependence in shrinkage. 

\begin{figure}
 \centering
    \includegraphics[width = 0.9\linewidth]{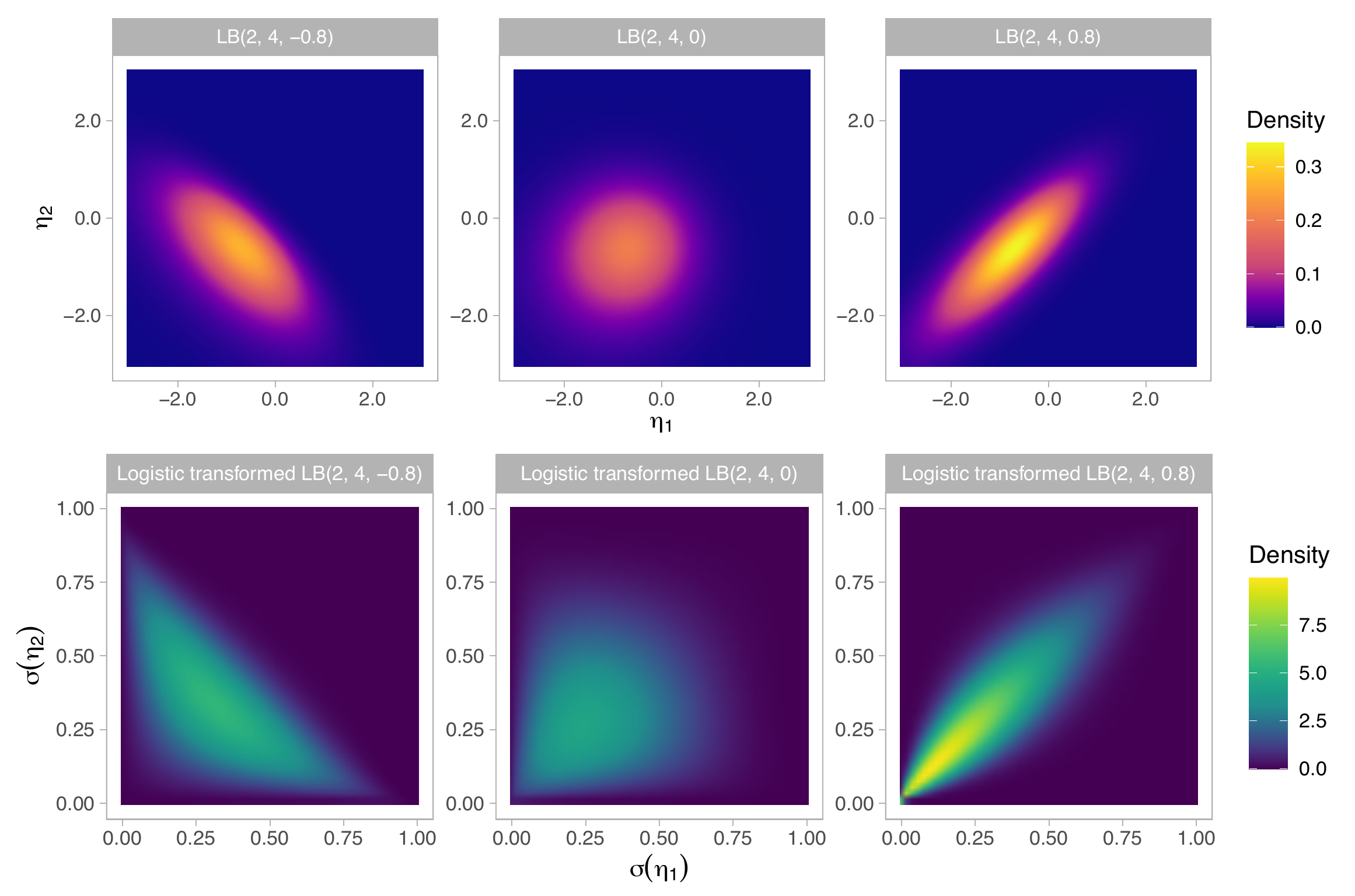}
\caption{Density plots of bivariate logistic-beta distributions with different correlation parameters $R_{12}\in\{-0.8,0,0.8\}$ and their logistic transformations. The logistic transformed random variables all have beta marginals with shape parameters $a=2$ and $b=4$.}
\label{fig:multivariatelb}
\end{figure}

The main advantage of this construction is that the correlation parameter $\bfR$ explicitly controls dependence with great flexibility.
In Proposition~\ref{prop:mlbmoments} and Corollary~\ref{coro:mlbcorrrange}, we present the first two moments of the multivariate logistic-beta and the range of pairwise correlation, illustrating how $\bfR$ controls the strength of linear dependence. 

\begin{proposition}
\label{prop:mlbmoments}
Let $\bm\eta \sim \lb(a,b,\bfR)$ with $R_{ij}$ being the $(i,j)$th element of $\bfR$. Then, we have
$
E(\eta_{i}) = \psi(a)-\psi(b)$, $\var(\eta_i) = \psi'(a) + \psi'(b)$ for $i=1,\dots,n$, and
\begin{equation*}
    \cov(\eta_i,\eta_j) = 
    \begin{cases} 
        2\psi'(a)R_{ij}, & \text{if } a=b,\\
    \psi'(a) + \psi'(b) + 2(R_{ij}-1)\{\psi(a) - \psi(b)\}/(a-b), &  \text{if } a\neq b,
    \end{cases}
    \end{equation*}
    where $\psi(x)$ and $\psi'(x)$ denote the digamma and the trigamma functions, respectively.
\end{proposition}

\begin{corollary}[Correlation range] 
\label{coro:mlbcorrrange}
Let $\bm\eta \sim \lb(a,b,\bfR)$ with some $2\times 2$ positive semidefinite correlation matrix $\bfR$. If $a=b$, $\corr(\eta_1,\eta_2)$ has a full range $[-1,1]$, and if $a\neq b$, the range of $\corr(\eta_1,\eta_2)$ is $[1-4\{\psi(a)-\psi(b)\}/[(a-b)\{\psi'(a)+\psi'(b)\}],1]$. 
\end{corollary}

From Proposition~\ref{prop:mlbmoments}, the covariance between $\eta_i$ and $\eta_j$ is a linear function of $R_{ij}$, aiding interpretation of $\bfR$. 
The P\'olya mixing variable $\lambda$ is shared across dimensions so that $R_{ij} = 0$ does not imply independence between $\eta_i$ and $\eta_j$. The covariance between the logistic transformed variables  $\cov\{\sigma(\eta_i),\sigma(\eta_j)\}$  lacks an analytical form, but Monte Carlo estimates are easily obtained.
While Corollary~\ref{coro:mlbcorrrange}
shows that arbitrarily high positive 
correlation can be captured for any choice of $(a,b)$, we have a nontrivial lower bound when $a\neq b$. 
The restriction on the strength of negative correlation is due to the shared asymmetric marginals \citep{Joe2006-ix}.

\subsection{Logistic-beta processes and correlation kernels}
We define a logistic-beta process on a discrete or continuous domain $\scrX$.
\begin{definition}
Let $\calR:\scrX\times \scrX\to [-1,1]$ be a correlation kernel which is positive semidefinite with $\calR(x,x) = 1$ for any $x\in \scrX$.   
We say $\{\eta(x) \in \mathds{R}: x\in \scrX\}$ is a logistic-beta process with shape parameters $a,b>0$ and correlation kernel $\calR$, denoted as $\eta(x)\sim \lbp(a,b,\calR)$, if every finite collection $\{\eta(x_1),\dots,\eta(x_n)\}^\T$ follows an $n$-dimensional multivariate logistic-beta with shape parameters $a,b$ and correlation parameter $\bfR$ with $(i,j)$th element $R_{ij} = \calR(x_i,x_j)$.
\end{definition}

From the definition of multivariate logistic-beta \eqref{eq:mlbnormalvariancemeanmixture} that shares a common P\'olya mixing variable and properties of the multivariate normal, 
the logistic-beta process is closed under marginalization and thus is a valid stochastic process as for 
other mixtures of Gaussian processes \citep{Yu2007-ej}. Specifically when $a=b$, it belongs to the family of elliptical processes \citep{Shah2014-kc}. The logistic transformation $\sigma\{\eta(x)\}$ has marginals distributed as $\mathrm{Beta}(a,b)$, with dependence induced by the correlation kernel $\calR$. In the following, we provide examples of correlation kernel constructions.

The logistic-beta process can capture groupwise, temporal, and spatial dependence in discrete or continuous domains.
When $\scrX$ corresponds to discrete time indices, we can consider an autoregressive order one correlation kernel $\calR(x, x') = \rho^{|x-x'|}$ for $|\rho|<1$. 
When $\scrX$ is a continuous spatial domain, a prominent example is the Mat\'ern kernel $\calR_M(x,x'; \varrho, \nu) = 2^{1-\nu}\Gamma(\nu)^{-1}(\|x-x'\|/\varrho)^\nu K_\nu(\|x-x'\|/\varrho)$ with parameters $\varrho, \nu>0$, where $K_\nu$ is the modified Bessel function of the second kind; see Figure~\ref{fig:lbp_matern} for realizations after logistic transformation. 
In Figure S.1.2 of the Appendix, we illustrate how $\corr\{\eta(x), \eta(x')\}$ and $\corr[\sigma\{\eta(x)\}, \sigma\{\eta(x')\}]$ change as a function of $\|x-x'\|$ under Mat\'ern correlation kernels and different choices of $(a,b)$.

\begin{figure}
\includegraphics[width=\textwidth]{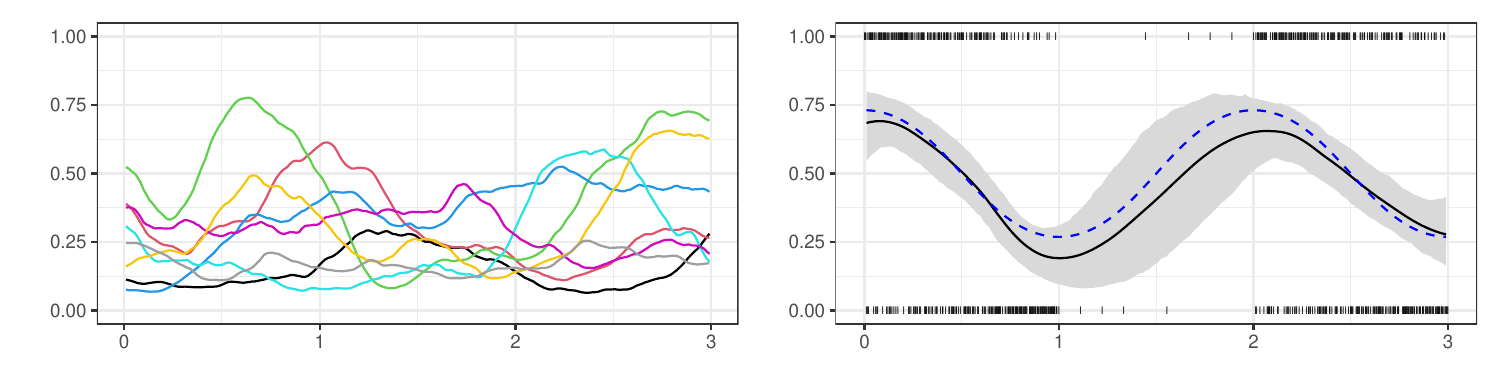}
    \caption{Illustration of the logistic-beta process on $\scrX=[0,3]$ with Mat\'ern correlation kernel. (Left) Eight realizations of $\sigma\{\eta(x)\}$ with $\eta(x)\sim \lbp(2,4,\calR_M)$, following $\mathrm{Beta}(2,4)$ marginally. (Right) Posterior mean (solid) and 95\% credible interval of latent success probabilities given the binary data $z(x_i)\indsim \mathrm{Ber}[\sigma\{\eta(x_i)\}]$, $i=1,\dots,n$ represented as small vertical bars. The dashed line corresponds to the true success probabilities.}
    \label{fig:lbp_matern}
\end{figure}

We can consider a more general construction of a correlation kernel $\calR$. Let $\bm\phi:\scrX \to \mathds{R}^q$ be a normalized feature map $\bm\phi(x) = \{\phi_1(x),\dots,\phi_q(x)\}^\T$ such that $\|\bm\phi(x)\|_2^2 =1$ \citep{Graf2003-by}. 
We assume a finite $q$, although infinite extensions are possible.
This normalized feature map can be obtained from an unnormalized feature map $\bm\varphi(x)$ with a bounded norm by letting $\bm\phi(x) =\bm\varphi(x)/\|\bm\varphi(x)\|_2$; 
for example, $\bm\varphi(x)$ can be spline basis functions for bounded domain $\scrX$, which induces non-linear dependence in $x$. Then, a valid correlation kernel $\calR$ can be defined as an inner product 
$\calR(x,x') = \langle \bm\phi(x),\bm\phi(x') \rangle.$
Such kernels lead to convenient representations of the logistic-beta process presented in Proposition~\ref{prop:lbpkernel_representations}, which are helpful in understanding the induced dependence structure and facilitating posterior inference. 

\begin{proposition} 
\label{prop:lbpkernel_representations}
Consider a normalized feature map $\bm\phi:\scrX\to \mathds{R}^q$ and the corresponding kernel $\calR(x,x') = \langle \bm\phi(x),\bm\phi(x') \rangle$. Let $\bm\Phi = \{\phi_k(x_i)\}_{ik}$ be an $n\times q$ basis matrix with $i$th row $\bm\phi(x_i)$. Then, a realization $\bm\eta = \{\eta(x_1),\dots,\eta(x_n)\}^\T$ from $\lbp(a,b,\calR)$ can be equivalently written as the linear predictor representation \eqref{eq:featuremap_linpred} and the hierarchical representation \eqref{eq:featuremap_hier},
\begin{align}
\bm\eta &= \{\psi(a)-\psi(b)\}(\bm{1}_n-\bm\Phi\bm{1}_q) + \bm\Phi \bm\beta, \quad \bm\beta \sim \lb(a,b,\bfI_q),\label{eq:featuremap_linpred}\\
\bm\eta &= 0.5\lambda (a-b)\bm{1}_n + \lambda^{1/2} \bm\Phi\bm\gamma,\quad   \lambda \sim \mathrm{Polya}(a,b),  \quad \bm\gamma\sim N_{q}(\bm{0}_q,\bfI_q). \label{eq:featuremap_hier}
\end{align}
\end{proposition} 

Representation \eqref{eq:featuremap_linpred} resembles basis function representations of a Gaussian process with normal priors on coefficients \citep[][\S 21.1]{Gelman2015-os}, but with a $q$-dimensional multivariate logistic-beta prior on the coefficients.  
When we consider an independent Bernoulli response model with logit link $z(x_i)\sim \mathrm{Ber}[\sigma\{\eta(x_i)\}]$, equation \eqref{eq:featuremap_linpred} corresponds to the linear predictor of the logistic regression model with an additional fixed varying intercept term. 
The hierarchical model representation \eqref{eq:featuremap_hier} is useful for posterior inference. Conditional on $\lambda$, the logistic-beta process with the normalized feature map kernel is parameterized with normal coefficients $\bm\gamma$, enabling conditionally conjugate updates.

\section{Latent logistic-beta process model for nonparametric binary regression}
\label{sec:3}
\subsection{Model}
\label{sec:3.1}

We consider a nonparametric binary regression problem in which logit-transformed probabilities of binary responses have a logistic-beta process prior. We highlight advantages of this prior in simplifying posterior inference and prediction while facilitating extensions. 
Let $z(x)\in\{0,1\}$ be binary data indexed by $x\in \scrX$. To induce 
 a $\mathrm{Beta}(a,b)$ prior for $\pr\{z(x) = 1\}$ for each $x$,
while enabling inference on   
probabilities at arbitrary  $x^*$ values, we choose  
\begin{align}
z(x_i) \mid \eta(x_i) &\indsim \mathrm{Ber}[\sigma\{\eta(x_i)\}] \quad (i=1,\dots,n), \quad \eta(x) \sim \lbp(a,b,\calR). 
    \label{eq:lbp_bernoulli}
\end{align}
By logistic-beta properties, the prior probabilities are marginally $\mathrm{Beta}(a,b)$ distributed, providing a natural extension from the independent beta-Bernoulli model to a dependent process model. 
The right panel of Figure~\ref{fig:lbp_matern} shows
an example of posterior prediction.  We simulated 600 binary data points with $\pr\{z(x) = 1\}= \sigma\{\cos(\pi x)\}$, and fitted the latent logistic-beta process model \eqref{eq:lbp_bernoulli} with $a=2,b=4$ and Mat\'ern correlation kernel with fixed parameters $(\varrho, \nu) = (0.3, 1.5)$. 
The middle part is slightly biased towards zero and has a wider credible interval, reflecting shrinkage towards the prior and uncertainty due to sparse data. 

\subsection{Posterior computation}
\label{section:3.2}
We describe a Markov chain Monte Carlo
algorithm for model \eqref{eq:lbp_bernoulli}.
The key components are P\'olya-Gamma data augmentation \citep{Polson2013-gb} and the normal variance-mean mixture representation \eqref{eq:mlbnormalvariancemeanmixture}, which leads to a latent Gaussian model conditional on the P\'olya mixing variable $\lambda$.
In Algorithm~\ref{alg:lb-bernoulli}, we describe a blocked Gibbs sampler. We use notations 
$\bfz = \{z(x_1),\dots,z(x_n)\}^\T$,  
$\bm\omega =  (\omega_1,\dots,\omega_n)^\T$ and $\bm\Omega = \diag(\bm\omega)$ for P\'olya-Gamma variables, $\bm\eta=(\eta_1,\dots,\eta_n)^\T = \{\eta(x_1),\dots,\eta(x_n)\}^\T$ and $\bfR = \{ \calR(x_i,x_j)\}_{i,j}$. Here we assume that the correlation kernel $\calR$ is known for conciseness, but incorporating unknown correlation kernel parameters is straightforward via an additional sampling step.

\begin{algorithm}[t]
\small
\caption{One cycle of a blocked Gibbs sampler for model \eqref{eq:lbp_bernoulli}.}
 {[1]} Sample $\omega_i \mid \eta_i \indsim \mathrm{PG}(1, \eta_i)$ for $i=1,\dots,n$, the P\'olya-Gamma distribution. \vspace{2mm}
\\
 {[2]} Sample $\lambda$ from $\pi(\lambda \mid \bm\omega, \bfz)$, where $\bm\eta$ is marginalized out,
 \begin{equation}
 \pi(\lambda\mid \bm\omega, \bfz )\propto \pipo(\lambda; a,b)\times N_n\big\{ \bm\Omega^{-1}(\bfz - 0.5\bm{1}_n); 0.5\lambda(a-b)\bm{1}_n, \lambda \bfR + \bm\Omega^{-1}\big\}.
 \label{eq:polyaupdate}
 \end{equation}
  {[3]} Sample $\bm\eta$ from $\pi(\bm\eta \mid \lambda, \bm\omega, \bfz)\propto \pi(\bfz \mid \bm\eta, \bm\omega)\times \pi(\bm\eta\mid \lambda)$, which is multivariate normal. 
\label{alg:lb-bernoulli}
\end{algorithm}

Step 1 of Algorithm~\ref{alg:lb-bernoulli} updates P\'olya-Gamma auxiliary variables $\bm\omega$ using a rejection sampler \citep{Polson2013-gb}.
Data augmentation with P\'olya-Gamma auxiliary variables converts the Bernoulli likelihood \eqref{eq:lbp_bernoulli} into a normal likelihood in terms of logit probabilities $\bm\eta$ conditional on P\'olya-Gamma variables $\bm\omega$.
Leveraging \eqref{eq:mlbnormalvariancemeanmixture}, $\pi(\bm\eta \mid \lambda)$ is also a multivariate normal in terms of $\bm\eta$.
This leads to joint conditional $\pi(\lambda, \bm\eta \mid \bm\omega, \bfz)\propto \pi(\bfz\mid \bm\eta, \bm\omega)\pi(\bm\eta\mid \lambda)\pi(\lambda)$ that is multivariate normal in terms of $\bm\eta$. 
Thus, the computation of collapsed conditional $\pi(\lambda \mid \bm\omega,\bfz) = \int \pi(\lambda, \bm\eta \mid \bm\omega, \bfz) \rmd \bm\eta$ can integrate out $\bm\eta$ by exploiting normal-normal conjugacy, greatly improving mixing behavior. 
Furthermore, again by conjugacy, Step 3 of Algorithm~\ref{alg:lb-bernoulli} updates $\bm\eta$ from a multivariate normal; for example when correlation matrix $\bfR$ is full rank,  
\[
\bm\eta \mid \bm\omega, \lambda, \bfz \sim N_n\big[(\bm\Omega + \lambda^{-1}\bfR^{-1})^{-1}\{ (\bfz - 0.5\bm{1}_n)+ 0.5(a-b) \bfR^{-1}\bm{1}_n \}, (\bm\Omega + \lambda^{-1}\bfR^{-1})^{-1} \big].
\]

Thus, the logistic-beta process offers significant computational benefits compared to multivariate beta constructions that rely on 
coordinate-wise updating 
\citep{Trippa2011-sz,DeYoreo2018-zs} or Metropolis-Hastings in high-dimensional space \citep{Arbel2016-or}. 

Sampling $\lambda$ from a collapsed conditional in Step 2 remains nontrivial and simple random walk Metropolis-Hastings would involve the calculation of P\'olya density $\pipo$ that is numerically unstable to evaluate. Based on Proposition~\ref{prop:polya_dens_identity}, we propose a novel class of adaptive P\'olya proposals that bypasses the evaluation of the P\'olya density. 
Denote $\calL(\lambda) = N_n\big\{ \bm\Omega^{-1}(\bfz - 0.5\bm{1}_n); 0.5\lambda(a-b)\bm{1}_n, \lambda \bfR + \bm\Omega^{-1}\big\}$, which is proportional to $\int \pi(\bfz\mid \bm\eta, \bm\omega)\pi(\bm\eta\mid \lambda) \rmd\bm\eta$ in terms of $\lambda$. The adaptive P\'olya proposal scheme selects the proposal distribution as $\mathrm{Polya}(a',b')$ with suitably chosen $(a',b')$ pair such that $a'+b' = a+b$. Then, the acceptance ratio for a candidate $\lambda^\star \sim \mathrm{Polya}(a',b')$ becomes $\min(1,\alpha_{\mathrm{MH}})$ with 
\[
\alpha_{\mathrm{MH}} =\frac{\pipo(\lambda; a',b')}{\pipo(\lambda^\star; a',b')} \frac{\pipo(\lambda^\star; a,b)\calL(\lambda^\star)}{\pipo(\lambda; a,b)\calL(\lambda)} = \exp\{(\lambda - \lambda^\star)(ab-a'b')/2\}\frac{\calL(\lambda^\star)}{\calL(\lambda)}
\]
where the P\'olya density terms are simplified by Proposition~\ref{prop:polya_dens_identity}, bypassing direct evaluation of $\pipo$. In terms of choosing the pair $(a',b')$ in an adaptive fashion, we use the moment matching method with a running average of $\lambda$. See Section S.2.2 of the Appendix for details.

When $n$ is large, sampling from $n$-dimensional multivariate normal in Step 3 becomes computationally expensive. If the correlation kernel is chosen based on $q \ll n$ normalized feature maps, $\bfR = \bm\Phi \bm\Phi^\T$ is a low-rank matrix, and the hierarchical representation \eqref{eq:featuremap_hier} in Proposition~\ref{prop:lbpkernel_representations} can be employed. 
Then, instead of sampling $n$-dimensional multivariate normal, Step 3 becomes sampling $\bm\gamma$ from $q$-dimensional multivariate normal with covariance matrix $\bfV_{\bm\gamma} = (\bfI_q + \lambda \bm\Phi^\T\bm\Omega\bm\Phi)^{-1}$ and mean $\bfV_{\bm\gamma}\lambda^{1/2}\bm\Phi^\T\{\bfz-0.5\bm{1}_n-0.5\lambda(a-b)\bm\omega\}$ and setting $\bm\eta = 0.5\lambda(a-b)\bm{1}_n + \lambda^{1/2}\bm\Phi\bm\gamma$, greatly reducing computational burden. In Section S.2.1 of the Appendix, we present further computational strategies based on scalable Gaussian process methods that preserve marginal variances.

The Algorithm~\ref{alg:lb-bernoulli} can be further extended to settings where beta parameters $(a,b)$ are assumed to be unknown and a hyperprior is placed on $(a,b)$. The same blocking strategy can be employed by jointly sampling $(a,b,\lambda)$ from the collapsed conditional using particle marginal Metropolis-Hastings \citep{Andrieu2010-tm}. See Section S.2.3 of the Appendix for details.

\section{Logistic-beta dependent Dirichlet processes}
\subsection{Definition and prior support}
\label{sec:4.1}
The latent logistic-beta process model can be easily employed in more complex Bayesian nonparametric models. In this section, we introduce a logistic-beta dependent Dirichlet process for modeling dependent random probability measures. Random probability measures $\{G_x:x\in \scrX \}$ follow a logistic-beta dependent Dirichlet process (LB-DDP) with concentration $b$, correlation kernel $\calR$, and independent atom processes $\{\theta_h(x):x\in \scrX\}$, $h\in\mathds{N}$ with $G^0_x$ the base measure, if 
\begin{align}
    G_{x}(\cdot) = \sum_{h=1}^{\infty}\Big(\sigma\{\eta_h(x)\}\prod_{l<h}\left[1-\sigma\{\eta_l(x)\}\right]\Big)\delta_{\theta_h(x)}(\cdot), \,\,\,\, \eta_h(x)\iidsim \lbp(1, b, \calR), \label{eq:lb-ddp}
\end{align}
for $h\in\mathds{N}$. 
Since the stick-breaking ratios $V_h(x) := \sigma\{\eta_h(x)\}$ have $\mathrm{Beta}(1,b)$ marginals for any $x\in \scrX$ and are independent across $h=1,2\dots$, the random probability measure $G_x$ for any given $x\in \scrX$ marginally follows a Dirichlet process with concentration parameter $b$. 
The hyperparameter $b$ has a familiar interpretation, aiding in elicitation. 

The prior introduced in equation~\eqref{eq:lb-ddp} inherits many of the amenable properties common to dependent Dirichlet processes. For instance, under mild conditions, the topological support of the LB-DDP coincides with the space of collections of all probability measures whose support is included in the support of base measure indexed by $x$, as the following result shows.
\begin{theorem}
\label{thm:lbddpsupport}
    Consider an LB-DDP with an atom process $\{\theta(x):x\in \scrX\}$ with support $\Theta$ that can be represented with a collection of copulas with positive density w.r.t. Lebesgue measure. Then, the logistic-beta dependent Dirichlet process has full weak support.
\end{theorem}
\noindent The above statement follows from direct applications of the results described in \citet{Barrientos2012-bx} and is a key requirement for posterior consistency. In the next subsection, we investigate further properties unique to the LB-DDP prior in equation~\eqref{eq:lb-ddp}.

\subsection{Prior dependence properties}

To gain a deeper understanding of the dependence structure induced by the LB-DDP, in this section we provide a full characterization of two key quantities: (i) the probability that two realizations $\vartheta\sim G_x$ and $\vartheta'\sim G_{x'}$ from two discrete random probability measures are equal, which we refer to as the  \emph{tie probability}, and (ii) the correlation between $G_x$ and $G_{x'}$ evaluated at common set $B$. 
We begin by considering the single-atoms LB-DDP, where $\theta_h(x)\equiv \theta_h \iidsim G^0$ in \eqref{eq:lb-ddp} so that atoms do not vary across $x$. Such an example allows us to describe the dependence induced by the weights alone.

\begin{theorem}
\label{thm:lbddpcorr}
Consider a single-atoms LB-DDP $\{G_x:x\in\scrX\}$ with concentration parameter $b$, correlation kernel $\calR$, and non-atomic base measure $G^0$. 
Let $\mu(x,x') = E[\sigma\{\eta(x)\} \sigma\{\eta(x')\}]$ with $\eta(x)\sim \lbp(1,b, \calR)$. Then, for any realizations $\vartheta\sim G_x$ and $\vartheta'\sim G_{x'}$, it holds that
\begin{equation}
\label{eq:lbddptie}
\pr(\vartheta = \vartheta'\mid G_x,G_{x'}) ={(1+b)}/\left\{2\mu(x,x')^{-1}-(1+b)\right\}.
\end{equation}
Moreover, for any Borel set $B$, we have
\begin{equation}
\label{eq:lbddpcorr}
\corr\{G_{x}(B),G_{x'}(B)\} ={(1+b)^2}/\left\{2\mu(x,x')^{-1}-(1+b)\right\}. 
\end{equation}
\end{theorem}
The tie probability in equation~\eqref{eq:lbddptie} provides an intuitive summary of the dependence between $G_x$ and $G_{x'}$ under general single-atoms DDP. Further multiplying such a quantity by $1+b$, one also obtains the correlation between random probability measures in equation~\eqref{eq:lbddpcorr}. The quantities in~\eqref{eq:lbddptie} and~\eqref{eq:lbddpcorr} attain the maximum of $1/(1+b)$ and $1$, respectively, when $x = x'$. Interestingly, their minimum is achieved when $\mu(x,x')$ is minimized, though the value depends on $b$ due to the nontrivial correlation lower bound described in Corollary~\ref{coro:mlbcorrrange}. 
We emphasize that the choice of $b$ affects the minimum possible tie probability and correlation, and within that range, the tie probability and correlation can be controlled by the correlation kernel $\calR$. In Figure S.1.2 in the Appendix, we illustrate how the correlation in equation~\eqref{eq:lbddpcorr} changes as a function of the distance between $x$ and $x'$ under the Mat\'ern kernel.

\begin{figure}
\centering
\includegraphics[width=\textwidth]{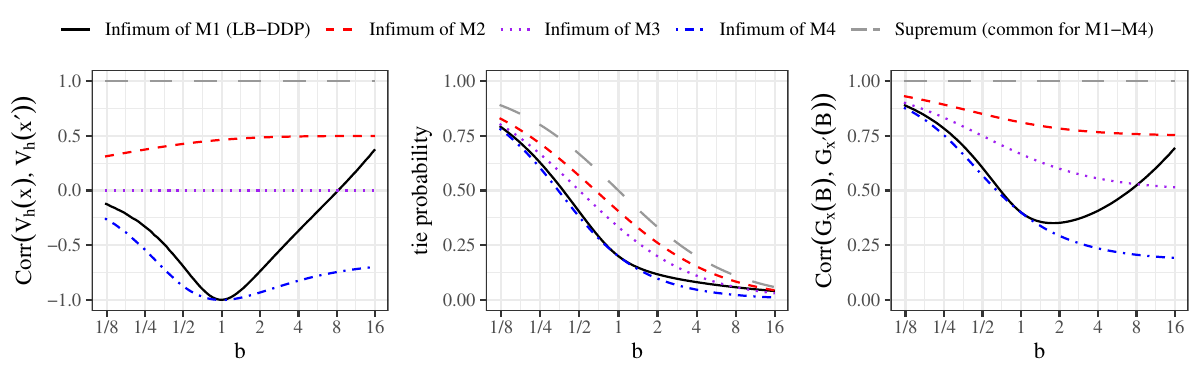}
    \caption{Analysis of the flexibility ranges induced by single-atoms DDPs (M1)--(M4) against the parameter $b$ (displayed in $\log_2$ scale). (Left) Ranges of $\corr\{V_h(x),V_h(x')\}$, the correlation between $\mathrm{Beta}(1,b)$ stick-breaking variables. (Center) Ranges of tie probability \eqref{eq:lbddptie}. (Right) Ranges of $\corr\{G_{x}(B),G_{x'}(B)\}$ \eqref{eq:lbddpcorr}.
    }
    \label{fig:ddpcorrrange}
\end{figure}
 
To investigate the improved flexibility of the LB-DDP prior, we compare the greatest lower bounds of correlation between stick-breaking ratios, tie probability \eqref{eq:lbddptie} and correlation between random probability measures \eqref{eq:lbddpcorr} against other DDP constructions with different multivariate beta distributions. Specifically, we consider four different single-atoms dependent Dirichlet process models based on dependent beta stick-breaking ratios: (M1) LB-DDP; (M2) The process proposed by \citet[][]{DeYoreo2018-zs}; (M3) The process proposed by \citet[][]{Nieto-Barajas2012-fr}; (M4) Copula-based dependent Dirichlet process. 
For M3, the minimum tie probability and correlation is achieved when $V_h(x)$ and $V_h(x')$ are independent of each other. For M4, the minimum is achieved with counter-monotonic beta random variables; examples include Gaussian copula with perfect negative correlation. 
 For completeness, we add brief descriptions of M2--M4 and their correlation lower bounds in Section S.1.3 of the Appendix.

Figure~\ref{fig:ddpcorrrange} illustrates how the range of the correlation between stick-breaking ratios, denoted as $\corr\{V_h(x),V_h(x')\}$, affects the ranges of tie probabilities and correlation between the random measure. We remark that, due to single-atoms assumption, $\corr\{G_{x}(B),G_{x'}(B)\}$ does not depend on $B$ and is always positive regardless of $\corr\{V_h(x),V_h(x')\}$; see \citet{Ascolani2024-cw}. The copula-based formulation M4 achieves the highest level of flexibility, which is an expected result due to Sklar's theorem \citep{Joe2014-ui}. However, inference with M4 often relies on Metropolis-Hastings in potentially high-dimensional space \citep{Arbel2016-or}, and the gap between M1 and M4 in Figure~\ref{fig:ddpcorrrange} corresponds to the cost of expressiveness that single-atoms LB-DDP pays to achieve computational tractability.
The behavior of the minimum tie probability and correlation of M1 varies distinctly under different values of $b$ compared to others. 
For small to moderate $b$, M1 can capture a much wider range of dependence compared to M2 and M3 by allowing negatively correlated stick-breaking ratios. 
When $b$ is large, the range of M1 becomes limited due to the nontrivial lower bound of correlation described in Corollary~\ref{coro:mlbcorrrange}.  This suggests that when  $b > 8$, the single-atoms assumption needs to be relaxed if one wants to ensure maximum flexibility relative to the independent weight model (infimum under M3). We formalize this in the following statement, providing a more general result of Theorem~\ref{thm:lbddpcorr}.

\begin{theorem}
\label{thm:lbddpcorr2}
    Consider the same settings as Theorem~\ref{thm:lbddpcorr}, but the single-atoms assumption is relaxed with i.i.d. (across $h$) atom processes $\{\theta_h(x):x\in\scrX\}$ with $G^0$ marginals for any $x$. Let $G^0_{x,x'}(B) = \pr(\theta_h(x)\in B, \theta_h(x')\in B)$ which does not depend on $h$. Then, 
\begin{equation}
    \corr\{G_{x}(B),G_{x'}(B)\} = \rho_0(B)\times{(1+b)^2}/\left\{2\mu(x,x')^{-1}-(1+b)\right\},
    \label{eq:lbddpcorr-general}
\end{equation}
where $\rho_0(B) = \{G^0_{x,x'}(B) -G^0(B)^2\}/\{G^0(B) - G^0(B)^2\}$. Moreover, $\rho_0(B)\in [-1,1]$.
\end{theorem}
\noindent For example, when the atom process is positively correlated but not identical across $x$, we have $G^0(B)^2 < G^0_{x,x'}(B) < G^0(B)$ so that $\corr\{G_x(B),G_{x'}(B)\}$ is multiplied by a fractional factor $\rho_0(B)$ compared to the single-atoms DDP case, alleviating the restricted range of correlation when $b$ is large.

In addition to the study of correlation structure, ensuring the continuity of the process as a function of covariates is important in the development of dependent Bayesian nonparametric processes \citep{MacEachern2000-lw}. In Theorem~\ref{thm:lbddpcontinuity}, we show that the continuity of correlation kernel $\calR$ at the diagonal leads to the continuity of the single-atoms LB-DDP as a function of a covariate in terms of correlation between random probability measures.

\begin{theorem} 
\label{thm:lbddpcontinuity}
Consider a single-atoms LB-DDP same as the setting of Theorem~\ref{thm:lbddpcorr}, and let $\calX$ be a compact subset of $\bbR^d$. If the correlation kernel $\calR(x,x')$ is continuous at every diagonal point $x=x'
$, then $\corr\{G_x(B),G_{x'}(B)\} \to 1$ as $x\to x'$.
\end{theorem}

\subsection{Logistic-beta dependent Dirichlet process mixture model}
We model a covariate-dependent response density as $f(y\mid x) = \int \calK(y\mid \theta) \mathrm{d}G_{x}(\theta)$, where $\{G_x\}$ follows a LB-DDP and $\calK(\cdot\mid \theta)$ is a mixture component distribution parametrized by $\theta$. For example, when we choose a normal kernel $\calK(y\mid \theta) = N_1(y; \mu, \tau^{-1})$ and atom processes $\theta_h(x) = \{\mu_h(x), \tau_h(x)\}$, $\mu_h(x) = \beta_{0h} + \beta_{1h}x$, $\tau_h(x)=\tau_h$ with semi-conjugate base measure, the density regression model becomes
\begin{align}
    f(y_i\mid x_i) &= \sum_{h=1}^\infty \Big(\sigma\{\eta_h(x_i)\}\prod_{l<h}[1-\sigma\{\eta_l(x_i)\}]\Big) N_1(y_i; \mu_h(x_i), \tau_h^{-1}), \label{eq:lbdddpmixturelikelihood} \\
    \eta_h(x) &\iidsim \lbp(1,b,\calR), \,\,\,\, (\beta_{0h},\beta_{1h})^\T\iidsim N_2(\bm{0}_2,\bm\Sigma_{\beta}), \quad \tau_h\iidsim \mathrm{Ga}(a_\tau,b_\tau) \label{eq:lbddpmixture2}
\end{align}
for $i=1,\dots,n$ and $h\in\mathds{N}$, where logistic-beta processes in \eqref{eq:lbddpmixture2} and atom processes both induce dependence on $x$, leading to a mixture of normal linear regression models with covariate-dependent weights. 

Algorithm~\ref{alg:lb-ddp} describes the blocked Gibbs sampler to estimate the parameters of the above model. 
To adapt Algorithm \ref{alg:lb-bernoulli} to model \eqref{eq:lbdddpmixturelikelihood}-\eqref{eq:lbddpmixture2}, we rely on the continuation-ratio logit representation of the multinomial distribution  
\citep{Tutz1991-jw}, which is also used in other stick-breaking process mixture models \citep{Rodriguez2011-pn,Rigon2021-ir}. 
Denoting the mixture component allocations as $\{s_i\}_{i=1}^n$ and 
$\scrZ_h = \{z_{ih}: s_i>h-1\}$ as a set of binary variables such that $z_{ih} = 1$ if $s_i=h$ and $z_{ih}=0$ if $s_i>h$, $\scrZ_h$ corresponds to observations of the latent logistic-beta process model for binary data with indices $I_h = \{i:s_i>h-1 \}$, independently across $h=1,2,\dots$.
We work on a truncated version having maximum number of components $H$ with $\sigma\{\eta_H(\cdot)\}\equiv 1$; if all the components are occupied we recommend increasing $H$.

\begin{algorithm}[t]
\caption{One cycle of the blocked Gibbs sampler for model \eqref{eq:lbdddpmixturelikelihood}--\eqref{eq:lbddpmixture2}.}
\label{alg:lb-ddp}
\small
 {[1]} Sample component allocations $s_i\in \{1,\dots,H\}$ for $i=1,\dots,n$, 
\[
\textstyle
\pr(s_i = h \mid - ) \propto \sigma\{\eta_h(x_i)\}\prod_{l=1}^{h-1}[1-\sigma\{\eta_l(x_i)\}]N_1(y_i ; \beta_{0h}+\beta_{1h}x_i, \tau_h^{-1}).
\]
 {[2]} For $h = 1,\dots,H-1$, with index set $I_h = \{i: s_i > h-1\}$, $\lambda_h$ and $\eta_h(\cdot)$, \\
 \hphantom{[0]} Run Algorithm~\ref{alg:lb-bernoulli} with binary data $\scrZ_h = \{\mathds{1}(s_i = h): i\in I_h $\}, update $\lambda_h$ and $\eta_h(\cdot)$. \vspace{2mm}\\
 {[3]} For $h=1,\dots,H$:
 \begin{enumerate}
     \item Sample component-specific regression parameters $(\beta_{0h},\beta_{1h}$) from
     \begin{equation*}
(\beta_{0h},\beta_{1h}\mid - )\, \indsim N_2\left(\tau_h\bfV_{h, \beta}\bfX_h^\T\mathbf{y}_h , \bfV_{h, \beta}\right),
 \end{equation*}
where $\bfV_{h, \beta} = (\tau_h\bfX_h^\T\bfX_h+\bm\Sigma_\beta^{-1})^{-1}$ and $\bfX_h$ and $\mathbf{y}_h$ are a matrix and a vector containing \\the rows $(1,x_i)$ and the scalars $y_i$ for which $s_i =h$, respectively.
 \item Sample component-specific precisions $\tau_h$ from
 \begin{equation*}
     \textstyle (\tau_h\mid - )\, \indsim \mathrm{Ga}\left(a_\tau + \sum_{i=1}^n\ind(s_i=h)/2, b_\tau + \sum_{i: s_i = h}(y_i-\beta_{0h}-\beta_{1h}x_i)^2/2\right).
 \end{equation*}
 \end{enumerate}
\end{algorithm}


\subsection{Extensions beyond dependent Dirichlet processes}
Logistic-beta processes can be further incorporated in dependent Bayesian nonparametric models beyond dependent Dirichlet processes. Consider two sequences $(a_h)_{h=1}^\infty$, $(b_h)_{h=1}^{\infty}$ such that $\sum_{h=1}^\infty \log(1+a_h/b_h)$ diverges to positive infinity \citep{Ishwaran2001-rd}. Then, given the correlation kernel $\calR$ and independent atom processes $\{\theta_h(x):x\in \scrX\}$, $h\in\mathds{N}$ with $G^0_x$ the base measure, one can consider the general class of logistic-beta dependent stick-breaking processes as
\begin{align}
    G_{x}(\cdot) = \sum_{h=1}^{\infty}\Big(\sigma\{\eta_h(x)\}\prod_{l<h}\left[1-\sigma\{\eta_l(x)\}\right]\Big)\delta_{\theta_h(x)}(\cdot), \quad \eta_h(x)\indsim \lbp(a_h, b_h, \calR), \label{eq:lb-dpyp}
\end{align}
for $h\in\mathds{N}$, using independent but non-identical logistic-beta processes. 
For example, if $a_h = 1-\varsigma$ and $b_h = b+h\varsigma$ for $\varsigma \in [0,1)$ and $b>-\varsigma$, this becomes a logistic-beta dependent Pitman-Yor process. 
That is, marginally for each $x$, $G_x$ follows a Pitman-Yor process with parameters $(\varsigma,b)$, and $\varsigma= 0$ reduces to LB-DDP. 

When considering mixture models, processes as in equation~\eqref{eq:lb-dpyp} may be preferred over a DDP mixture for situations when data exhibit strong heterogeneity across $x$ in terms of the number of clusters; see \citet{Bassetti2014-fa} for a detailed comparison. In such cases, the above dependent Pitman-Yor process construction based on logistic-beta provides computational benefits, since the inference algorithm is essentially the same as Algorithm~\ref{alg:lb-ddp}. Finally, estimation for the hyperparameters $a_h$ and $b_h$ is each component can be suitably performed via particle marginal Metropolis-Hastings without direct evaluation of the P\'olya density, as in the standard LB-DDP mixture model; see Section S.2.3 of the Appendix for a description.


\section{Empirical demonstrations}
\label{sec:5}

\subsection{Nonparametric binary regression simulation study}
\label{sec:5.1}
Our simulation studies on nonparametric binary regression have two main objectives. Firstly, we aim to assess the benefits of posterior computation strategies involved in latent logistic-beta process models. We focus on how the blocking strategy in the Gibbs sampler and adaptive P\'olya proposal affect mixing. 
Secondly, we compare the latent logistic-beta process with Gaussian copula models, which also give dependent random probabilities with beta marginals. 
We consider a nonparametric binary regression problem on a spatial domain  $\scrX = [0,1]^2$, where spatial coordinates $\{x_i\}_{i=1}^n$ are generated uniformly at random from $\scrX$.

We considered the latent logistic-beta process model \eqref{eq:lbp_bernoulli} with $(a,b) = (1,2)$ and the Gaussian copula model as two different data generation models. 
The Gaussian copula model first draws a sample $\big\{\zeta(x_1),\dots, \zeta(x_n)\big\}^\T$ from a mean zero unit variance Gaussian process and then applies transformation $\zeta \mapsto F_B^{-1}\{F_Z(\zeta)\}$ to yield beta marginals, where $F_Z$ is the standard normal c.d.f. and $F_B^{-1}$ is the inverse c.d.f. of the Beta($1,2$) distribution. 
For both the latent logistic-beta process and Gaussian copula models, we used Mat\'ern correlation kernels with fixed smoothness $\nu = 1.5$.
We considered different degrees of spatial dependence $\varrho\in\{0.1, 0.2, 0.4\}$, resulting in six different data generation scenarios.
Given the success probabilities, binary data $\{z(x_i)\}_{i=1}^n$ of size $n=500$ were independently sampled, with $n_{\mathrm{test}}=100$ data points reserved to evaluate predictive performance. 
We simulated 100 replicated datasets for each data generation scenario.

The latent logistic-beta process and Gaussian copula models are fitted to binary data with size $n_{\mathrm{train}} = 400$. 
Beta parameters are fixed at the true value $(a,b) = (1,2)$, and Mat\'ern correlation kernels are adopted with a fixed smoothness parameter $\nu = 1.5$. 
The Gibbs sampler for the latent logistic-beta process model is implemented in \texttt{R} \citep{R_Core_Team2023-wk}, and we explore four posterior inference algorithms, considering blocking and/or the adaptive P\'olya proposal scheme in Step 2 of Algorithm~\ref{alg:lb-bernoulli}. For the Gaussian copula model, we use \texttt{Stan}  \citep{Carpenter2017-lo} to implement Hamiltonian Monte Carlo. 
For both models, we run 2,000 Markov chain Monte Carlo iterations and record wall-clock running time, with the first 1,000 samples discarded as a burn-in. Posterior predictive distributions of success probabilities in the test dataset are compared with true values to assess predictive performance.

\begin{table}
\caption{Nonparametric binary regression simulation results comparing computational benefits of adaptive P\'olya proposal scheme in the blocked Gibbs sampler. Entries shown are averages over 100 replicates, with Monte Carlo standard errors shown in parentheses. ESS, effective sample size of $\lambda$; ESS/sec, ESS divided by running time in seconds.
}
\footnotesize
\centering
\begin{tabular}{cc c c c}
    \toprule
       Data generation & Algorithm settings & ESS & ESS/sec & Acceptance rate of $\lambda$ (\%)\\
    \midrule
    \multirow{2}{*}{\makecell{LBP, \\ $\varrho = 0.1$}} & Adapted & 245.08 (12.86) & 3.35 (0.18) & 54.28 (1.03) \\
     & Non-adapted & 196.35 (11.28)  & 2.69 (0.15) & 49.39 (1.39) \\
     \midrule
    \multirow{2}{*}{\makecell{LBP, \\ $\varrho = 0.2$}} & Adapted & 257.01 (16.32) & 2.89 (0.18) &  62.26 (1.12) \\
      & Non-Adapted &  247.83 (16.53) &  2.82 (0.19) & 57.60 (1.47)\\
     \midrule
    \multirow{2}{*}{\makecell{LBP, \\ $\varrho = 0.4$}} & Adapted & 368.26 (17.54) & 4.12 (0.20) & 66.12 (0.98) \\
      & Non-adapted & 328.72 (19.00) & 3.67 (0.21) & 61.45 (1.32) \\
    \bottomrule
  \end{tabular}

  \label{table:sim1}
\end{table}

Table~\ref{table:sim1} reports the benefits of the adaptive P\'olya proposal scheme based on the mixing behavior of $\lambda$. Here ``non-adapted'' corresponds to the fixed independent Metropolis-Hastings proposal distribution of $\lambda$ using prior distribution $\mathrm{Polya}(a,b)$. 
We consider effective sample size (ESS) based on 1,000 posterior samples, effective sample rate (ESS/sec), and an acceptance rate of $\lambda$ from an independent Metropolis-Hastings sampler where a higher acceptance rate is desirable \citep[][\S 7.6.1]{Robert2004-de}. 
The adaptive P\'olya proposal scheme improves the acceptance rate of $\lambda$ in all scenarios. 
In Section S.3.1 of the Appendix, we also show that the blocking scheme in Steps 2 and 3 in Algorithm~\ref{alg:lb-bernoulli}, which is made possible due to the normal mixture representation of the logistic-beta, has a substantial impact on the mixing behavior, leading to significant improvement on the effective sample size of $\lambda$.

\begin{table}
\caption{Nonparametric binary regression simulation results comparing latent logistic-beta process (LBP) and Gaussian copula models (GauCop). 
  Entries are averaged over 100 replicates, with Monte Carlo standard errors in parentheses. RMSE, root mean squared error; CRPS, continuous ranked probability score; mESS/sec, multivariate effective sample size of latent success probabilities divided by running time in seconds.}
 \footnotesize
 \centering
  \begin{tabular}{ccc cc cc}
    \toprule
   \multirow{2}{*}{\makecell{Data\\ generation}} & \multirow{2}{*}{Model} & \multicolumn{2}{c}{RMSE $\times 100$} & \multicolumn{2}{c}{mean CRPS $\times 100$} & \multirow{2}{*}{mESS/sec}  \\
   \cmidrule(lr){3-4} \cmidrule(lr){5-6}
   &  & training & test & training & test & \\
     \midrule
    \multirow{2}{*}{\makecell{GauCop, \\ $\varrho = 0.1$}} & LBP & 11.93 (0.14) & 12.32 (0.17) & 6.59 (0.10) & 6.80 (0.11) & 21.11 (0.21)\\
     &GauCop & 11.82 (0.13) & 12.24 (0.16) & 6.48 (0.09) &6.71 (0.11) & 0.48 (0.01)\\
    \midrule
    \multirow{2}{*}{\makecell{GauCop, \\ $\varrho = 0.2$}} & LBP &8.67 (0.15) & 8.80 (0.16) & 4.78 (0.10) & 4.85 (0.10) & 17.58 (0.16) \\
     & GauCop &  8.61 (0.16) &  8.75 (0.17) & 4.74 (0.10) & 4.82 (0.11) & 0.40 (0.01)\\
    \midrule
    \multirow{2}{*}{\makecell{GauCop, \\ $\varrho = 0.4$}} & LBP & 6.11 (0.16) & 6.14 (0.16) & 3.39 (0.10) & 3.41 (0.10) & 17.41 (0.17) \\
     &GauCop &  6.10 (0.16) & 6.13 (0.16) & 3.38 (0.10) & 3.40 (0.10)  & 0.47 (0.01)\\
    \bottomrule
  \end{tabular}
  
    \label{table:sim2}
\end{table}

Table~\ref{table:sim2} reports the in-sample and predictive performance of the latent logistic-beta process and Gaussian copula models, as well as the sampling efficiency with Gaussian copula data generation scenarios. In terms of root mean square error and mean continuous ranked probability score of latent success probabilities across both training and test data, the models produce nearly identical results. Although data were generated from the Gaussian copula model, the latent logistic-beta process model is able to recover latent probabilities as accurately as the true model. 

Regarding sampling efficiency of latent success probabilities measured by multivariate effective sample size \citep[][]{Vats2019-wo} divided by the algorithm's running time in seconds, the latent logistic-beta process model dramatically outperforms the Gaussian copula model due to the significantly longer running time of the Gaussian copula model. 
Although the implementation of the Gaussian copula model with the \texttt{Stan} program is based on C++ code with highly optimized algorithms, the computation of the log-posterior and its gradient at each step demands high computational costs. The algorithm of the latent logistic-beta process model is currently implemented using \texttt{R} software, which has the potential for improvement in running time with a lower-level programming language. 
Section S.3.1 of the Appendix contains further details on simulation settings and additional simulation results under different data generation scenarios. 

\subsection{Bayesian density regression simulation study}

We now assess the density regression performance of the LB-DDP mixture model using simulated data. Specifically, we compare the LB-DDP with the linear DDP \citep{De-Iorio2009-kz}, which shares the same model as \eqref{eq:lbdddpmixturelikelihood}-\eqref{eq:lbddpmixture2} but $\sigma\{\eta_h(x_i)\}$ are replaced with $V_h\iidsim \mathrm{Beta}(1,b)$. This is a single-weights DDP where the only atoms depend on $x$. Hence, the difference between the two lies in the logistic-beta process prior alone, allowing us to evaluate the impact of its additional flexibility explicitly. 
We also compare LB-DDP with the logit stick-breaking process \citep[LSBP,][]{Ren2011-el}. While it is not a DDP, the LSBP shares similar characteristics with LB-DDP by placing logistic transformed latent Gaussian processes in stick-breaking ratios, effectively acting as surrogate for a dependent Dirichlet process with covariate-dependent weights.

We consider two different data generation scenarios from two previous studies, which we refer to as Scenario A \citep{Dunson2007-wj} and Scenario B \citep{Wade2025-ef}, respectively. The true conditional densities are defined as
\begin{align}
    p_A(y\mid x_i) &= \exp(-2x_i)N_1(y; x_{i}, 0.1^2) + \{1-\exp(-2x_i)\}N_1(y; x_i^4, 0.2^2),\\
p_B(y\mid x_i) &= 
\begin{cases} 
N_1(y; 0, 0.2^2), & -2 \le x_i \le 2,\\
N_1(y; 2x_i-4, 0.05^2),& 2 < x_i \le 5,\\
N_1(y; 6, (x_i-5)^2/15 + 0.01),& 5 < x_i \le 10,\\
\end{cases}\
\end{align}
and we generated covariates from $x_i\iidsim \mathrm{Unif}(0,1)$ for Scenario A and $x_i\iidsim \mathrm{Unif}(-2,10)$ for Scenario B.  With two sample sizes $n\in\{500,2000\}$, we simulated 100 replicated datasets for both scenarios.

In each Scenario, we adopt the following modeling choices. For the correlation kernel $\calR$ in LB-DDP, we choose a normalized feature map kernel based on a natural cubic spline basis with $6$ degrees of freedom $\{\phi_k(x)\}_{k=1}^6$, defined on $[0,1]$ for Scenario A and $[-2,10]$ for Scenario B. The same set of basis functions are used in LSBP, by letting $\sigma^{-1}\{V_h(x)\} = \eta_h(x) = \bm\phi(x)^\T\bm\alpha_h$ and $\bm\alpha_h\sim N_q(\bm{0}_q, \sigma_\alpha^2 \bfI_q)$. This is in line with previous stick-breaking process literature \citep{Rigon2021-ir, Horiguchi2024-jh, Wade2025-ef}. 
For the hyperparameter choice, we consider the case when $b\in \{0.2, 1, 2\}$ for linear DDP and LB-DDP, and $\sigma_\alpha^2 \in\{10^2, \pi^2/3, 0.2^2\}$ for LSBP. These three different hyperparameter choices result in similar tie probabilities \eqref{eq:lbddptie} with $x=x'$ (the gray dash in center panel of Figure~\ref{fig:ddpcorrrange}), which are approximately $0.84, 0.50$, and $0.34$, respectively. These tie probabilities do not depend on $x$ under normalized feature maps. Our indicator of performance are the regression error $E[\{\widehat{E(y\mid x)}-E(y\mid x)\}^2]^{1/2}$ and the density error $E\{\int |\widehat{p(y\mid x)} - p(y\mid x)| \rmd y\}$ averaged across all test covariates $x$, similar to \citet{Wade2025-ef}. See Section S.3.2 of the Appendix for further details.

\begin{figure}
    \centering
    \includegraphics[width=\linewidth]{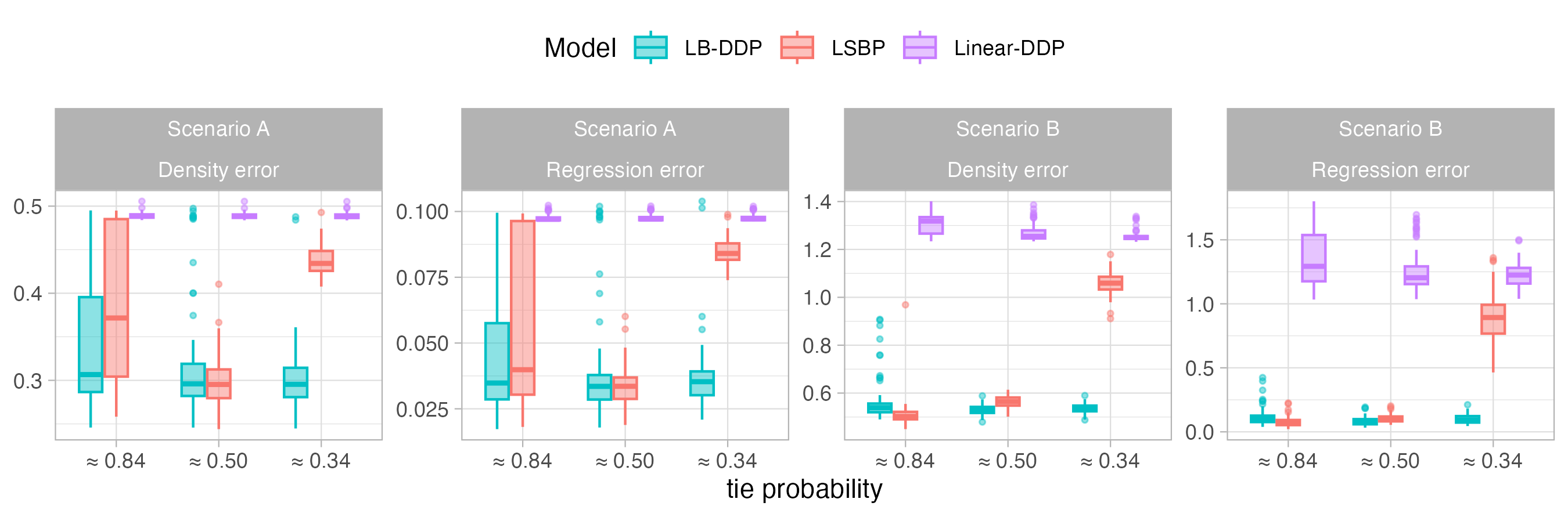}
    \caption{Density regression simulation results under two different data generation scenarios with sample size $n=500$ and three different hyperparameter settings based on prior tie probability. Density error and regression error quantify the accuracy of conditional density and conditional mean estimates averaged across test covariates, respectively.}
    \label{fig:sim2boxplotn500}
\end{figure}

Figure~\ref{fig:sim2boxplotn500} reports the result when $n=500$; boxplots for $n=2000$ are available in Figure S.3.1 of the Appendix. Compared to the linear DDP mixture model, the performance of the LB-DDP and LSBP mixture model are significantly improved in both scenarios, which shows a clear advantage of the flexibility offered by covariate-dependent weights. Meanwhile, compared to LB-DDP, the results of the LSBP mixture model are highly sensitive to the choice of hyperparameters, even though they are calibrated to induce similar prior tie probabilities. 
Thus, consistent with findings of \citet{Wade2025-ef}, the choice of hyperparameters in stick-breaking process priors can drastically affect performance, partly due to an unknown covariate-dependent marginal structure of prior that can abruptly change in certain regions of $x$. 
Overall, the density regression simulation result shows the clear benefits of LB-DDP in terms of flexibility, interpretability, and robustness to the choice of hyperparameters. 

\subsection{Application to pregnancy outcome data analysis}
\label{sec:5.2} 
We use the logistic-beta dependent Dirichlet process mixture model \eqref{eq:lbdddpmixturelikelihood}--\eqref{eq:lbddpmixture2} to analyze DDE exposure and pregnancy outcome data \citep{Longnecker2001-il}. DDE is a metabolite of the chemical compound DDT that was widely used as an insecticide.  
We analyze how the conditional distribution of gestational age of delivery $y$ changes with the level of exposure to DDE $x$, focusing on the probability of preterm birth, defined as gestational age of delivery less than 37 weeks. The dataset is publicly available in R package \texttt{BNPmix} \citep{Corradin2021-ag}. For illustration, we focus on the $n=1023$ women who smoked during pregnancy; results for other groups are in Section S.3.3 of the Appendix.

We compare the result with the LSBP mixture model in terms of the estimated probability of preterm birth given DDE exposure $\pr(y\le 37\mid x)$ as a function of $x$. For the choice of correlation kernel $\calR$ for LB-DDP and basis functions for LSBP, similar to \citet{Rigon2021-ir}, we consider normalized feature map kernels obtained from a natural cubic spline basis with 6 degrees of freedom. In both models, the atom processes are specified as in
\eqref{eq:lbdddpmixturelikelihood}, leading to a mixture of linear regressions with covariate-dependent weights. For the choice of hyperparameters, we consider the same three settings $b\in \{1/5, 1, 2\}$ as in the simulation study, denoted setting 1, 2, 3, respectively, where the LSBP hyperparameters are chosen such that the prior tie probabilities closely match the proposed model within each setting.

Posterior computation is performed with Gibbs samplers for both models. We obtained 35,000 iterations and the first 5,000 samples were discarded as burn-in. Both algorithms were implemented in \texttt{R} and had similar running times; 7.5 minutes for the LB-DDP mixture model and 5.6 minutes for the LSBP mixture model in setting 2.

\begin{figure}
\centering
    \includegraphics[width=\textwidth]{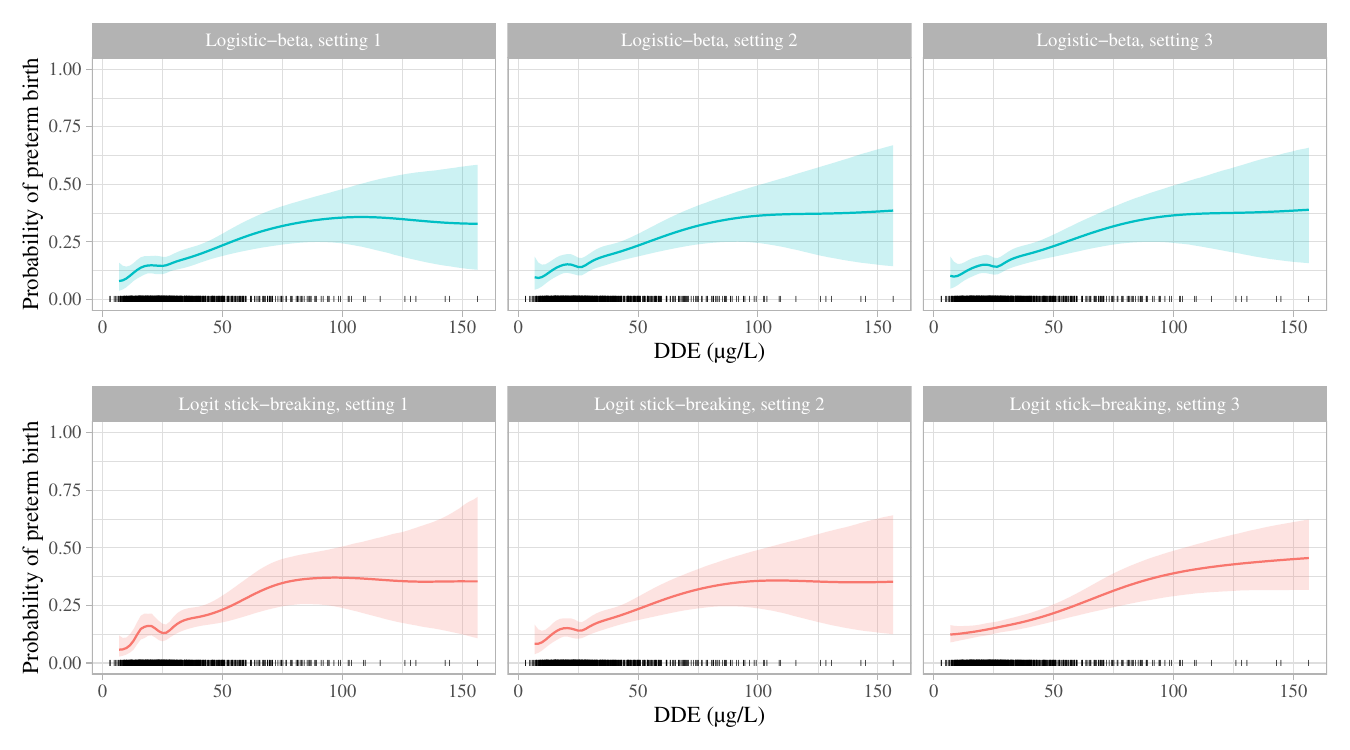}
    \caption{Estimated probability of preterm birth with 95\% credible intervals against DDE exposure level from the LB-DDP mixture model (top) and the LSBP mixture model (bottom) with three different hyperparameter settings corresponding to prior tie probabilities approximately 0.84, 0.50, and 0.34. DDE data are represented by small vertical bars.}
    \label{fig:ddp_pretermbirth_smoke}
\end{figure}

The estimated preterm birth probabilities given DDE exposures are summarized in Figure~\ref{fig:ddp_pretermbirth_smoke}. For both models, the probabilities
tend to increase with DDE levels, with wider credible intervals at high levels because of the increasingly sparse data. There is a stark difference between the two models in terms of how the estimates and credible intervals vary according to hyperparameter choice. Although each hyperparameter setting yields similar prior properties and both models use the same basis functions to incorporate dependence in $x$, Figure~\ref{fig:ddp_pretermbirth_smoke} shows that the posterior estimate of the LSBP is highly sensitive compared to the proposed model. 
Specifically, the hyperparameter choice of the LSBP model has a substantial impact on uncertainty estimates at high DDE exposure levels, where setting 1 and setting 3 of the LSBP model have much wider and narrower credible intervals, respectively, compared to the proposed model. 
This sensitivity, as noted in \citet{Wade2025-ef}, makes conclusions drawn from the LSBP model questionable, whereas the proposed LB-DDP model is much more robust to the choice of hyperparameters. 

Section S.3.3 of the Appendix contains details on pregnancy outcome data analysis settings and additional results with the non-smoking group of size $n=1290$ that gives the same conclusion. 

\section{Discussion}
\label{sec:6}

Although we have focused on two particular applications, logistic-beta processes can be usefully integrated within many other Bayesian nonparametric models. A particularly intriguing case is to incorporate covariate dependence into Indian buffet processes, relying on a beta-Bernoulli representation. An interesting application is to species distribution modeling in ecology \citep{Elith2009-pi}, where the features correspond to species and covariates correspond to environmental data, such as location or temperature. 
Such a model can accommodate and predict new species discovery, which is common in ecological studies of insect and fungi.

One promising future avenue of research is to develop a new family of stochastic processes having possibly different beta marginal distributions, while maintaining computational benefits when used for modeling dependent random probabilities. In the logistic-beta process, the common beta marginal property is due to the shared P\'olya mixing parameter $\lambda$. Such a new family could be used to develop even more flexible dependent Bayesian nonparametric models, such as DDP models with a covariate-dependent concentration parameter.

\section*{Software}
Codes for reproducing all analyses are available at the following Github repository: \url{https://github.com/changwoo-lee/logisticbeta-reproduce}.

\section*{Acknowledgements}

This project has received funding from U.S. National Science Foundation (DMS-2210456, DMS-2220231, IIS-2426762), U.S. National Institutes of Health (R01ES035625), U.S. Office of Naval Research (N00014-21-1-2510-P00004-7), and from the European Research Council (ERC) under the European Union’s Horizon 2020 research and innovation programme (856506).

\bibliographystyle{apalike}
\bibliography{paperpile.bib}

\newpage

\appendix
\begin{center}
    \Large \bf Appendices
\end{center}

\setcounter{section}{0}
\setcounter{equation}{0}
\setcounter{table}{0}
\setcounter{proposition}{0}
\renewcommand{\theequation}{S.\arabic{equation}}
\renewcommand{\thesection}{S.\arabic{section}} 
\renewcommand{\theproposition}{S.\arabic{proposition}} 
\renewcommand{\thetheorem}{S.\arabic{theorem}} 
\renewcommand{\thecorollary}{S.\arabic{corollary}} 
\renewcommand{\thelemma}{S.\arabic{lemma}} 
\renewcommand{\thetable}{S.\arabic{table}} 
\renewcommand{\thefigure}{S.\arabic{figure}}

\counterwithin{figure}{section}
\counterwithin{table}{section}
\counterwithin{proposition}{section}
\counterwithin{theorem}{section}

\section{Proofs and additional propositions}
\label{appendix:a}
\subsection{Properties of P\'olya distributions}
\label{appendix:a1}
The P\'olya distribution serves as a mixing distribution in the normal variance-mean mixture representation of logistic-beta, which includes standard logistic distribution as a special case. Thus, the P\'olya distribution generalizes the squared Kolmogorov--Smirnov distribution \citep{Andrews1974-ms}, used as the scale mixing distribution in Gaussian scale mixture representations of the standard logistic distribution  \citep{Chen1998-fn,Holmes2006-np}. 

We collect results related to the P\'olya distribution, including its moments (Proposition~\ref{prop:polyameanvar}) and density identity (Proposition~\ref{prop:polya_dens_identity}).

\begin{proposition} 
\label{prop:polyameanvar}
The mean and variance of $\lambda\sim \pipo(a,b)$ are
\begin{equation*}
\label{eq:polyamean}
E(\lambda) = 
\begin{cases} 2\{\psi(a)-\psi(b)\}/(a-b), & \text{if }a \neq b, \\
2\psi'(a), & \text{if }a = b. 
\end{cases}    
\end{equation*}

\begin{equation*}
\label{eq:polyavar}    
\var(\lambda) = 
\begin{cases}
\frac{4}{(a-b)^2}\left[\psi'(a) + \psi'(b) - \frac{2}{(a- b)}\{\psi(a) - \psi(b)\}\right],& \text{if }a \neq b, \\
2\psi^{(3)}(a)/3, & \text{if }a = b. 
\end{cases}
\end{equation*}
where $\psi^{(3)}$ is a 3rd derivative of digamma function $\psi$.
\end{proposition}

\begin{proof}
First check the mean. When $a\neq b$,
    \begin{align*}
    E(\lambda) = \sum_{k=0}^\infty \frac{2}{(k+a)(k+b)}   &= \frac{2}{a-b}\sum_{k=0}^\infty\left(\frac{1}{k+b} - \frac{1}{k+a}\right)  = \frac{2\{\psi(a)-\psi(b)\}}{a-b},
    \end{align*}
    see \citet[][\S 5.7.6]{Olver2010-jj} using subtraction. When $a=b$, we have $\sum_{k=0}^\infty 2/(k+a)^2 = 2\psi'(a)$, see \citet[][\S 5.15.1]{Olver2010-jj}. 
    For variance, if $a\neq b$,
     \begin{align*}
    \var(\lambda) = \sum_{k=0}^\infty \frac{4}{(k+a)^2(k+b)^2}   &= \frac{4}{(a-b)^2}\sum_{k=0}^\infty \left\{\frac{1}{(k+a)^2} + \frac{1}{(k+b)^2} - \frac{2}{(k+a)(k+b)}\right\}\\
    &= \frac{4}{(a-b)^2}\left[\psi'(a) + \psi'(b) - \frac{2}{(a- b)}\{\psi(a) - \psi(b)\}\right]
    \end{align*}
    and when $a=b$, $\var(\lambda) = 4(-1)^{3+1}\psi^{(3)}(a)/3! = 2\psi^{(3)}(a)/3$; see \citet[][\S 25.11.12]{Olver2010-jj}.
\end{proof}

\subsubsection{Proof of Proposition~\ref{prop:polya_dens_identity}}

\begin{proof} 
Let $a,b>0$ be given and $(a',b')$ be a pair that satisfies $a'+b'=a+b$ and $a',b'>0$. Starting from \eqref{eq:polyadensity},  
\begin{align*}
       &\pipo(\lambda ; a', b') \\
       &= \sum_{k=0}^\infty(-1)^k \binom{a'+b'+k-1}{k}\frac{(k+(a'+b')/2)}{B(a',b')}\exp\left(-\frac{(k+a')(k+b')}{2}\lambda\right)\\
       &= \frac{B(a,b)}{B(a',b')}\sum_{k=0}^\infty \left\{(-1)^k   \binom{a+b+k-1}{k}\frac{(k+(a+b)/2)}{B(a,b)}\times \right.\\
       & \qquad \qquad\qquad \quad  \left.\exp\left(-\frac{(k+a)(k+b)}{2}\lambda + \frac{(ab-a'b')\lambda}{2}\right)\right\}\\
       &= \frac{B(a,b)}{B(a',b')}\exp\left(\frac{(ab-a'b')\lambda}{2}\right) \pipo(\lambda ; a, b), 
\end{align*}
which proves Proposition~\ref{prop:polya_dens_identity}.
\end{proof}

\subsection{Properties of multivariate logistic-beta and logistic-beta processes}
\label{appendix:a2}
We summarize results related to multivariate logistic-beta and logistic-beta processes, including the moments (Proposition~\ref{prop:mlbmoments}), correlation range (Corollary~\ref{coro:mlbcorrrange}), and representations of logistic-beta processes under normalized feature map kernel (Proposition~\ref{prop:lbpkernel_representations}).

\subsubsection{Proofs of Proposition~\ref{prop:mlbmoments} and Corollary~\ref{coro:mlbcorrrange}}

\begin{proof}
From \citet[][\S23.10]{Johnson1995-sk}, when $\eta$ is a univariate logistic-beta random variable with parameters $a$ and $b$, we have  
$E(\eta) = \psi(a)-\psi(b)$ and $\var(\eta) = \psi'(a)+\psi'(b)$. 
Let $\eta \sim \lb(a,b,R)$ with $R_{ij}$ being the $(i,j)$th element of $R$. By the law of total covariance, the covariance and correlation is given by
\begin{align*}
\cov(\eta_i,\eta_j) &= E\{\cov(\eta_i,\eta_j \mid \lambda)\} + \cov\{E(\eta_i \mid \lambda), E(\eta_j \mid \lambda)\} \\
&=  R_{ij}E(\lambda) + \cov\{0.5\lambda(a-b),0.5\lambda(a-b)\}\\
&= 
\begin{cases} 
        2\psi'(a)R_{ij}, & a=b,\\
    \psi'(a) + \psi'(b) + 2(R_{ij}-1)\left\{\frac{\psi(a) - \psi(b)}{a-b}\right\}, &  a\neq b.
    \end{cases}
\end{align*}
and 
\[
\corr(\eta_i,\eta_j) = \begin{cases} 
     R_{ij}, &a=b,\\
     1 + (R_{ij}-1)\left[\frac{2\{\psi(a) - \psi(b)\}}{(a-b)\{\psi'(a) + \psi'(b)\}}\right], &a\neq b.
     \end{cases}
\]
which shows Proposition~\ref{prop:mlbmoments}. 
Considering a bivariate case, if $a=b$,  the range of $\corr(\eta_1,\eta_2)$ is $[-1,1]$. If $a\neq b$, the upper bound 1 of $\corr(\eta_1,\eta_2)$ is attained at $R_{12} = 1$, and lower bound of $1-4\{\psi(a)-\psi(b)\}/\{(a-b)(\psi'(a)+\psi'(b))\}$ is attained at $R_{12} = -1$. This shows Corollary~\ref{coro:mlbcorrrange}. 
\end{proof}

\subsubsection{Proof of Proposition~\ref{prop:lbpkernel_representations}}
\begin{proof}
    Consider a normalized feature map $\bm\phi:\scrX\to \mathds{R}^q$ and the corresponding correlation kernel $\calR(x,x') = \langle \bm\phi(x), \bm\phi(x') \rangle$ for logistic-beta process. 
    Let $\bm\Phi = [\phi_k(x_i)]_{ik}$ be an $n\times q$ basis matrix with $i$th row corresponding to $\bm\phi(x_i)$. Then, for a realization $\bm\eta = (\eta(x_1),\dots,\eta(x_n))^\T$, we have
\begin{align*}
 \bm\eta \mid \lambda \sim N_n\left(0.5\lambda (a-b)\bm{1}_n,  \lambda\bm\Phi\bm\Phi^\T \right), \quad \lambda \sim \mathrm{Polya}(a,b).
\end{align*} 
However, $\bm\eta \mid \lambda$ is a rank deficient multivariate normal when $q<n$. By introducing $\bm\gamma \sim N_q(\bm{0}_q, \bfI_q)$, it can be equivalently written as 
\[
\eta = 0.5\lambda (a-b)\bm{1}_n + \lambda^{1/2} \bm\Phi\bm\gamma,\quad   \lambda \sim \mathrm{Polya}(a,b),  \quad \bm\gamma\sim N_{q}(\bm{0}_q,\bfI_q), 
\]
which corresponds to the hierarchical representation of Proposition~\ref{prop:lbpkernel_representations}. Also,
\begin{align*}
    \bm\eta &= 0.5\lambda (a-b)\bm{1}_n +  \lambda^{1/2} \bm\Phi\bm\gamma\\
    &= 0.5\lambda (a-b)\bm{1}_n - 0.5\lambda(a-b) \bm\Phi\bm{1}_q + 0.5\lambda(a-b)\bm\Phi\bm{1}_q +  \lambda^{1/2}\bm\Phi\bm\gamma\\
    &= 0.5\lambda(a-b)(\bm{1}_n - \bm\Phi\bm{1}_q) + \bm\Phi\{0.5\lambda (a-b)\bm{1}_q + \lambda^{1/2} \bm\gamma\}.
\end{align*}
Thus, if $\lambda \sim \mathrm{Polya}(a,b)$ is marginalized out, by the linearity of expectation and definition of multivariate logistic-beta, we have 
\[
\bm\eta = \{\psi(a)-\psi(b)\}(\bm{1}_n-\bm\Phi\bm{1}_q) + \bm\Phi \bm\beta, \quad \bm\beta\sim \lb(a,b,\bfI_q),
\]
where we used $E(\lambda) = 2\{\psi(a) - \psi(b)\}/(a-b)$ as in Proposition~\ref{prop:polyameanvar}. This shows the linear predictor representation of Proposition~\ref{prop:lbpkernel_representations}.
\end{proof}

\subsection{Logistic-beta dependent Dirichlet process and related models}
\label{appendix:a3}

We present results on the logistic-beta dependent Dirichlet process, including full weak support property (Theorem~\ref{thm:lbddpsupport}), correlation between random probability measures (Theorem~\ref{thm:lbddpcorr}), and continuity (Theorem~\ref{thm:lbddpcontinuity}).

\subsubsection{Proof of Theorem 4.1}


\begin{proof}
First consider the single-atoms logistic-beta dependent Dirichlet process, so that atoms $\theta_h$ are i.i.d. samples from the base measure $G_0$. Let $\eta(\cdot)\sim \lbp (1,b,\calR)$ and $F_B$ be the c.d.f. of $\mathrm{Beta}(1,b)$. Consider an $n$-dimensional continuous random variable 
$
\bfU = (F_B[\sigma\{\eta(x_1)\}],\dots,F_B[\sigma\{\eta(x_n)\}]) 
$
where its cumulative distribution $C_{x_1,\dots,x_n}$ is a copula since $\bfU$ has uniform marginals by construction. It is clear that $\bfU$ has a positive density w.r.t. Lebesgue measure since $F_B$ and $\sigma$ are all continuous and injective. 
 Now consider a collection of copula functions $\calC = \{C_{x_1,\dots,x_n}:x_1,\dots,x_n\in\scrX, n>1\}$. Since 
 $\sigma$ is the c.d.f. of a standard logistic random variable, 
 the collection of copula functions satisfies conditions of Corollary 1 of \citet{Barrientos2012-bx}. Also, the stick-breaking ratios of logistic-beta dependent Dirichlet process $\sigma\{\eta_h(x)\}$ are determined by the set of copulas $\calC$ and the set of beta marginal distributions with parameters $1$ and $b$. Therefore, the full weak support property of the logistic-beta dependent Dirichlet process follows from Theorem 3 of \cite{Barrientos2012-bx}. 
 For a general logistic-beta dependent Dirichlet process where atoms and weights both depend on $x$, the full weak support condition is satisfied as long as the atom process can be represented with a collection of copulas with positive density w.r.t. Lebesgue measure.

\end{proof}

\subsubsection{Proof of Theorem 4.2 and Theorem 4.3}

\begin{proof} 
We first show Theorem 4.2 and specialize to Theorem 4.1 later. The proof is an extended version of the Theorem 1 of \citet{Dunson2008-ds}. Denote $V_h = \sigma\{\eta_h(x)\}$ and $V_h' = \sigma\{\eta_h(x')\}$, which marginally follow $V_h\sim \mathrm{Beta}(1,b)$ and $V_h'\sim \mathrm{Beta}(1,b)$. Since the dependent Dirichlet process prior implies $G_x$ marginally follows a Dirichlet process with concentration parameter $b$, we have $E\{G_x(B)\} = G^0(B)$ and $\var\{G_x(B)\} = G^0(B)\{1 - G^0(B)\}/(1 + b)$ for any $x$ \citep{Muller2015-rt}. In addition to the notation $\mu(x,x') = E[\sigma\{\eta(x)\} \sigma\{\eta(x')\}]$ where $\eta(x)\sim \lbp(1,b,\calR)$, denote $\mu(x) = E[\sigma\{\eta(x)\}]$ which is simply $\mu(x) = 1/(1+b)$ for any $x$. Recall the notation $G^0_{x,x'}(B) = \pr(\theta_h(x)\in B, \theta_h(x')\in B)$, which does not depends on $h$ since atom processes are i.i.d. across $h$.

The mixed moment is 
\begin{align*}
  & E\{G_{x}(B)G_{x'}(B)\} \\
  =\,\, &  E\left[ \sum_{h=1}^\infty V_h\left\{\prod_{l=1}^{h-1}(1-V_l)\right\}\delta_{\theta_h(x)}(B) \times \sum_{h=1}^\infty V_h'\left\{\prod_{l=1}^{h-1}(1-V_l')\right\}\delta_{\theta_h(x')}(B) \right]\\
  =\,\, &G_{x,x'}^0(B) \underbrace{E\left[\sum_{h=1}^\infty  V_hV_h'\prod_{l=1}^{h-1}(1-V_l)(1-V_l')\right]}_{(*)} \\
  & + G^0(B)^2 \underbrace{E\left[ \sum_{h\ge 1, k \ge 1, h\neq k} V_hV_k'\prod_{l=1}^{h-1}(1-V_l)\prod_{m=1}^{k-1}(1-V'_m) \right]}_{(**)} 
\end{align*}
where equality follows from atom processes and weight processes being independent. Since $(V_h,V_h')$ is independent of $(V_k,V_k')$ for any $h\neq k$, the term $(*)$ becomes
\begin{align*}
(*) &= \sum_{h=1}^\infty  E(V_hV_h')\prod_{l=1}^{h-1}E\{(1-V_l)(1-V_l')\} \\
&= \sum_{h=1}^\infty \mu(x,x')\{1-\mu(x)-\mu(x')+\mu(x,x')\}^{h-1}\\
&= \frac{\mu(x,x')}{\mu(x)+\mu(x')-\mu(x,x')} = \frac{(1+b)}{2/\mu(x,x')-(1+b)}
\end{align*}
The term $(**)$ becomes, since $(V_h,V_h')$ is independent of $(V_k,V_k')$ for any $h\neq k$,
\begin{align*}
 (**) &= E\sum_{h=1}^\infty \sum_{k=1}^{h-1} V_h \left[\prod_{m=1}^{k-1}(1-V_m)(1-V_m')\right](V_k'-V_kV_k')\prod_{s=k+1}^{h-1}(1-V_s)\\
 &+ E\sum_{h=1}^\infty \sum_{k=h+1}^{\infty} V_k'\left[\prod_{l=1}^{h-1}(1-V_l)(1-V_l')\right](V_h-V_hV_h')\prod_{s=h+1}^{k-1}(1-V_s')\\
 &= \sum_{h=1}^\infty \sum_{k=1}^{h-1} \mu(x) \{1-\mu(x)-\mu(x')+\mu(x,x')\}^{k-1}\{\mu(x')-\mu(x,x')\}\{1-\mu(x)\}^{h-k-1}\\
 &+\sum_{h=1}^\infty \sum_{k=h+1}^{\infty} \mu(x')\{1-\mu(x)-\mu(x')+\mu(x,x')\}^{h-1}\{\mu(x)-\mu(x,x')\}\{1-\mu(x')\}^{k-h-1}\\
 &= \sum_{k=1}^\infty \sum_{h=k+1}^{\infty} \mu(x) \{1-\mu(x)-\mu(x')+\mu(x,x')\}^{k-1}\{\mu(x')-\mu(x,x')\}\{1-\mu(x)\}^{h-k-1}\\
 &+\sum_{h=1}^\infty \sum_{k=h+1}^{\infty} \mu(x')\{1-\mu(x)-\mu(x')+\mu(x,x')\}^{h-1}\{\mu(x)-\mu(x,x')\}\{1-\mu(x')\}^{k-h-1}\\
  &= \sum_{k=1}^\infty \{1-\mu(x)-\mu(x')+\mu(x,x')\}^{k-1}\{\mu(x')-\mu(x,x')\}\\
 &+\sum_{h=1}^\infty \{1-\mu(x)-\mu(x')+\mu(x,x')\}^{h-1}\{\mu(x)-\mu(x,x')\}\\
 &= \frac{\mu(x) +\mu(x') -2\mu(x,x')}{\mu(x) + \mu(x') -\mu(x,x')}
\end{align*}
where reordering is justified as the series is absolutely convergent. Combining, 
\[
E\{G_x(B)G_{x'}(B)\} = \frac{G_{x,x'}^0(B)\mu(x,x')}{\mu(x) + \mu(x') -\mu(x,x')} + G^0(B)^2\frac{\mu(x) +\mu(x') -2\mu(x,x')}{\mu(x) + \mu(x') -\mu(x,x')}
\]
and thus, using $E\{G_x(B)\} = E\{G_{x'}(B)\} = G^0(B)$ and $\mu(x) = \mu(x') = 1/(1+b)$,
\[
\cov\{G_x(B),G_{x'}(B)\} = \frac{\{G_{x,x'}^0(B) -G^0(B)^2\}\mu(x,x')}{2/(1+b) -\mu(x,x')}
\]
and using $\var\{G_x(B)\} = G^0(B)(1 - G^0(B))/(1+b)$,
\[
\corr\{G_x(B),G_{x'}(B)\} = \frac{G_{x,x'}^0(B) - G^0(B)^2}{G^0(B) - G^0(B)^2}\frac{(1+b)^2}{2/\mu(x,x')-(1+b)}
\]

Now consider Theorem 4.1, the single-atoms DDP. We have $\theta_h(x) = \theta_h(x')$ for any $x,x'$ and $h$ and thus $G_{x,x'}^0(B) = \pr(\theta_h(x)\in B, \theta_h(x')\in B) = G^0(B)$, so that first term cancels out and correlation does not depend on the evaluation set $B$, which gives 
\[
\corr\{G_x(B),G_{x'}(B)\} = {(1+b)^2}/\left\{2\mu(x,x')^{-1}-(1+b)\right\}.
\]
Since atoms are common in $G_x$ and $G_{x'}$, the tie probability is the same as the probability of picking the same label, corresponding to the term $(*)$ in the previous display. Thus, for $\vartheta\sim G_x$, $\vartheta'\sim G_{x'}$, we have
\[
\pr(\vartheta = \vartheta'\mid G_x,G_{x'}) ={(1+b)}/\left\{2\mu(x,x')^{-1}-(1+b)\right\}
\]
which completes the proof.
\end{proof}

\subsubsection{Proof of Theorem 4.4}

\begin{proof}
    If $\calR(x,x')$ is continuous at every point such that $x=x'$, the covariance function $\cov\{\eta(x),\eta(x')\}$ is a linear function of $\calR(x,x')$ (Proposition~\ref{prop:mlbmoments}) and thus continuous at every point $x=x'$, i.e. the logistic-beta process is mean square continuous \citep[][p.26]{Adler1981-at}. By \citet{Wise1977-bu}, its logistic transformation $V(\cdot) = \sigma\{\eta(\cdot)\}$ is also mean square continuous, which implies that $E\{V(x)V(x')\}\to E\{V(x')^2\}$ as $x\to x'$. Since $V(\cdot)$ has $\mathrm{Beta}(1,b)$ marginal distribution, $E\{V(x)^2\} = 2/\{(b+1)(b+2)\}$ for any $x\in\calX$. Therefore, $\corr\{G_x(B),G_{x'}(B)\}\to (1+b)^2/\{(b+1)(b+2)-(1+b)\} = 1$ as $x\to x'$.
\end{proof}

\subsubsection{Details on related models}

Next, we present the prior correlation structure under other various dependent Dirichlet process models. 
In addition to the logistic-beta dependent Dirichlet process (M1), we provide a derivation of correlation lower bounds under three different dependent Dirichlet process models: (M2) squared \textsc{AR}(1) process-based dependent Dirichlet process \citep{DeYoreo2018-zs}, (M3) time-series dependent Dirichlet process \citep{Nieto-Barajas2012-fr}, and (M4) copula-based dependent Dirichlet process. Similar to the logistic-beta dependent Dirichlet process, these models induce dependence through a stochastic process on stick-breaking ratios that have beta marginals but with different constructions. 

First, from \citet[][Appendix A.1.]{DeYoreo2018-zs}, the $\corr\{V(x),V(x')\}$ under the squared \textsc{ar}(1) process-based dependent Dirichlet process has the greatest lower bound of $b^{1/2}(b+1)(b+2)^{1/2} - b(b+2)$ when autoregressive parameter converges to 0. This implies $E\{V(x)V(x')\}$ has the infimum of $\{b^{3/2}+(1-b)(b+2)^{1/2}\}/\{(b+1)(b+2)^{1/2}\}$, using $E\{V(x)\} = 1/(b+1)$ and $\var\{V(x)\} = b/((b+1)^2(b+2))$ for any $x$ under beta distributed $V(x)$. This leads to the greatest lower bound of correlation between random probability measures $(1+b)/([2(b+2)^{1/2}/\{b^{3/2}+(1-b)(b+2)^{1/2}\}]- 1)$, which converges to $0.75$ as $b\to \infty$.

Next, the infimum of $\corr\{V(x),V(x')\}$ under the time-series dependent Dirichlet process \citep{Nieto-Barajas2012-fr} is when $V(x)$ and $V(x')$ are independent so that $\inf \corr\{V(x),V(x')\} = 0$. Thus, $E\{V(x)V(x')\}$ has the infimum of $1/(b+1)^2$, which leads to $\inf \corr\{G_{x}(B), G_{x'}(B)\} = (1+b)/(1+2b)$. This fact is also noted in \citet{Nieto-Barajas2012-fr} as a special case when stick-breaking ratios are independent, and the same fact also applies to the dependent Dirichlet process proposed by \citet{Taddy2010-ax} when the autoregressive parameter converges to 0.  

Finally, we analyze a copula-based single-atoms dependent Dirichlet process, with beta stick-breaking ratios constructed through a copula that induces a stochastic process with $\mathrm{Beta}(1,b)$ marginals. Finding infimum of $\corr\{G_{x}(B), G_{x'}(B)\}$ is equivalent to finding the infimum of correlation between stick-breaking ratios $\corr\{V(x), V(x')\}$, which is achieved at the Fr\'echet lower bound copula \citep{Embrechts2002-mr}. In other words, minimum correlation is attained when $V(x)$ and $V(x')$ satisfy $V(x') = F_B^{-1}(1 - F_B(V(x)))$ (counter-monotonic), where $F_B$ is the c.d.f. of $\mathrm{Beta}(1,b)$. 

Although the analytic form of $\corr\{V(x), V(x')\}$ under common atoms logistic-beta dependent Dirichlet process and copula-based dependent Dirichlet process are not readily available due to the nonlinear transformations involved, those can be estimated with Monte Carlo simulation where we used 5,000 simulated datasets of size 1,000 to calculate the correlation estimates.

\begin{figure}
\centering
\includegraphics[width=0.63\textwidth]{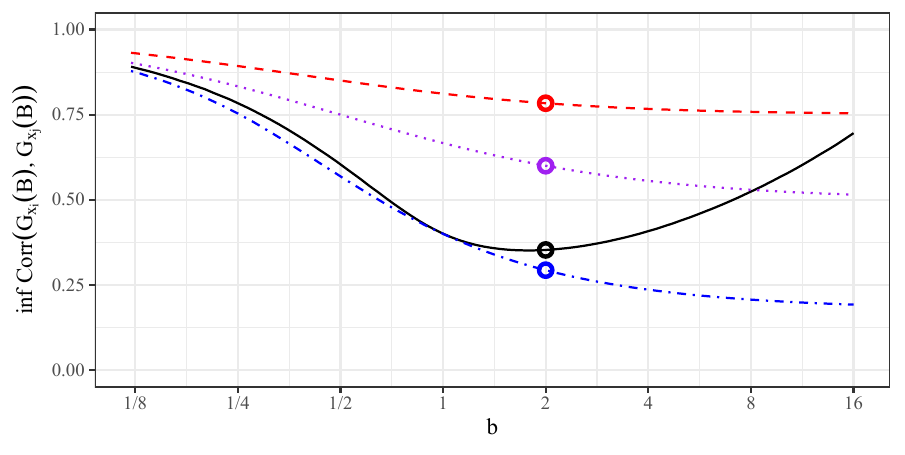}
\includegraphics[width=0.33\textwidth]{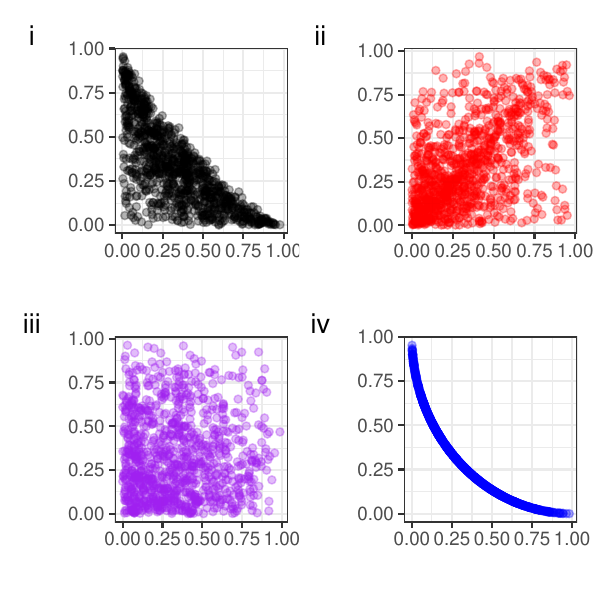}
    \caption{Illustration of different correlation lower bounds under four different single-atoms dependent Dirichlet process. (Left) Infimum of correlation between random probability measures under models M1 (solid, black), M2 (dashed, red), M3 (dotted, purple), and M4 (dot-dashed, blue). (Right) Scatterplot of bivariate betas with $\mathrm{Beta}(1,2)$ marginals, corresponding to circles of the left panel.}
    \label{fig:ddp_betas_scatterplot}
\end{figure}

Figure~\ref{fig:ddp_betas_scatterplot} illustrates different correlation lower bounds under four different single-atoms dependent Dirichlet processes with concentration parameter $b=2$. The difference in the induced dependence structure is clearly seen from scatterplots, where 1,000 random samples are drawn from a corresponding bivariate beta distribution that induces minimum correlation from models M1--M4. The Fr\'echet lower bound clearly shows counter-monotonicity between two beta random variables, which induces minimum possible correlation. Transformation of logistic-beta also can induce a negative correlation between beta random variables, while others cannot. In summary, logistic-beta dependent Dirichlet processes can capture a wide range of dependence when $b$ is small to moderate. 

In addition, Figure~\ref{fig:corr_dist} illustrates how various correlations change as a function of distance between $x$ and $x'$ under Mat\'ern correlation kernel with different smoothness parameters $\nu\in\{0.5, 1.5, 2.5\}$. This includes correlation between logistic-beta realizations $\corr\{\eta(x), \eta(x')\}$, correlation between their logistic transformations $\corr\{V(x), V(x')\}$ which are used as stick-breaking ratios, and correlation between random probability measures under single-atoms LB-DDP. In general, the correlation decay pattern of $\corr\{V(x), V(x')\}$ and $\corr\{G_x(B), G_{x'}(B)\}$ are similar to the decay pattern of correlation between logistic-beta realizations $\corr\{\eta(x), \eta(x')\}$, but in different scales depending on the choice of beta parameters $(a,b)$.

\begin{figure}
    \centering
    \includegraphics[width=\linewidth]{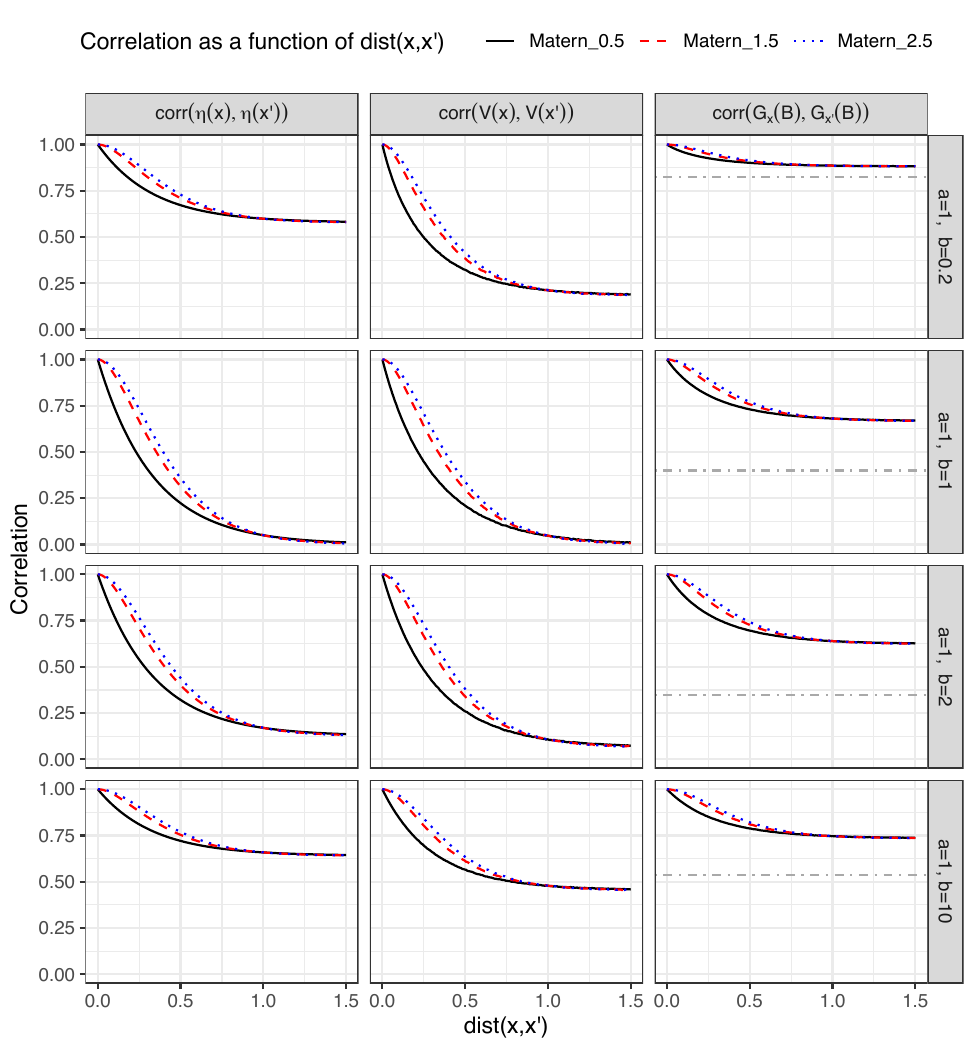}
    \caption{Correlation as a function of $d = \|x-x'\|$ under Mat\'ern kernel $\calR_M(x,x') = 2^{1-\nu}\Gamma(\nu)^{-1}(d/\varrho)^\nu K_\nu(d/\varrho)$ with $(\varrho,\nu) \in\{ (1/3, 1/2), (4/19, 3/2), (10/59, 5/2)\}$. Each row corresponds to parameters $(a,b)\in \{(1,1/5), (1,1),(1,2), (1,10)\}$. (Left) $\corr\{\eta(x),\eta(x')\}$, the correlation between logistic-beta process realizations. When $a=b$, we have $\corr\{\eta(x),\eta(x')\} = \calR_M(x,x')$; see Proposition~\ref{prop:mlbmoments}. (Center) $\corr\{V(x),V(x')\}$ with $V(x) = \sigma\{\eta(x)\}$, the correlation between stick-breaking ratios. (Right) $\corr\{G_x(B),G_{x'}(B)\}$, the correlation between random probability measures under the single-atoms logistic-beta dependent Dirichlet process. Dot-dashed horizontal lines (gray) correspond to $\inf\corr\{G_x(B),G_{x'}(B)\}$ when $\calR(x,x') = -1$. The range parameters $\varrho$ are chosen such that when $a=b$, $\corr\{\eta(x),\eta(x')\}$ at distance 1 becomes approximately 0.05 (effective range is 1).}
    \label{fig:corr_dist}
\end{figure}

\section{Additional posterior computation strategies}
\label{appendix:bnew}

\subsection{Correlation kernels with low-rank structures}
\label{appendix:b1}
Thanks to the normal variance-mean mixture construction of the logistic-beta process, many existing scalable Gaussian process methods that preserve marginal variances can be easily applied to the logistic-beta process. In addition to the correlation kernel constructed from normalized feature maps,  
another example of correlation kernel with low-rank structures is the modified predictive process \citep{Finley2009-bc}. 
From a parent kernel $\calR$ with $q$ knot locations $u_1, \dots,u_q \in\scrX$, a new correlation kernel $\tilde\calR$ with a low-rank structure can be defined as 
\begin{equation}
    \tilde{\calR}(x,x') = \bfr(x)^\T \bfR_{uu}^{-1}\bfr(x') + \mathds{1}(x=x') (1- \bfr(x)^\T \bfR_{uu}^{-1}\bfr(x'))
    \label{eq:R_mpp}
\end{equation}
where $\bfr(x) = [\calR(x,u_k)]_{k=1}^q$ is a length $q$ vector as a function of $x$ and a matrix $\bfR_{uu} = [\calR(u_k,u_{k'})]_{k,k'=1}^q$ is a fixed $q\times q$ knot correlation matrix; see also \citet{Quinonero-Candela2005-sz}. 
Thanks to the Woodbury matrix identity and matrix determinant lemma, posterior computation with a low-rank correlation kernel can be carried out in a highly efficient manner. This includes fast evaluation of $n$-dimensional multivariate normal density in Step 2 of Algorithm~\ref{alg:lb-bernoulli}, as the covariance matrix has a ``low-rank plus diagonal'' structure. 
In addition to the modified predictive process \eqref{eq:R_mpp}, other options include full-scale approximation \citep{Sang2012-vj} or multi-resolution approximation \citep{Katzfuss2017-lo} among many others.

Since marginal variances of a new kernel must remain 1 to be a valid correlation kernel, the predictive process \citep{Banerjee2008-iw} or Vecchia approximations \citep{Vecchia1988-gf} cannot be directly applied without suitable adjustments. Exploring other options for scalable logistic-beta processes is an interesting future direction. 

\subsection{Adaptive P\'olya proposal and particle Gibbs sampler}
\label{appendix:b2}

We describe the details of the adaptive P\'olya proposal scheme, on choosing the proposal distribution parameters $(a',b')$ and sampling $\lambda$ from a distribution that involves P\'olya prior density $\pipo(\lambda; a,b)$, such as expression \eqref{eq:polyaupdate} in Step 2 of Algorithm~\ref{alg:lb-bernoulli}. 
There are many ways to incorporate previous samples' information in the choice of $(a',b')$. An example is the moment matching method with a running average of $\lambda$, denoted as $\bar{\lambda}$, that is recursively updated as $\bar{\lambda} \leftarrow \bar{\lambda}(m-1)/m + \lambda/m$ in $m$th iteration of the algorithm, leading to adaptive Metropolis-within-Gibbs~\citep{Roberts2009-xo}. Assuming $a'\le b'$ without loss of generality due to the symmetry of the P\'olya distribution in its arguments, which implies $a'\le c/2$, the $\mathrm{Polya}(a',c-a')$ has a mean $h(a') = 2\{\psi(a') - \psi(c-a')\}/(2a'-c)
$ (or $h(a') = 2\psi'(c/2)$ if $a'=c/2$) that is decreasing on $a'\in (0,c/2]$. Thus, we set $a' = (a+b)/2$ if $\bar{\lambda}$ is less than or equal to $2\psi'((a+b)/2)$, and otherwise we set $a'$ to the solution of $h(x) = \bar{\lambda}$ which can be easily obtained with numerical optimization algorithms.

In Step 2 of Algorithm~\ref{alg:lb-bernoulli}, it is possible to use a particle Gibbs sampler \citep{Andrieu2010-tm} instead of Metropolis-Hastings to potentially improve the mixing behavior of $\lambda$. 
The adaptive P\'olya proposal scheme can be used to choose the particle proposal. In Step 2 of Algorithm~\ref{alg:lb-bernoulli}, instead of drawing a single candidate, the particle Gibbs sampler first draws $N$ candidates $\lambda^{(1)},\dots, \lambda^{(N)}\sim \mathrm{Polya}(a',b')$ and sets the current state as the last candidate $\lambda^{(N+1)} = \lambda$, where $(a',b')$ are chosen based on adaptive P\'olya proposal scheme. Then, unnormalized importance weights $w_k = \exp\{-\lambda^{(k)}(ab-a'b')/2\}\calL(\lambda^{(k)})$ for $k=1,\dots,N+1$ are calculated, and $\lambda$ is sampled among $N+1$ candidates according to probabilities proportional to $w_k$'s. This particle Gibbs sampler also avoids the evaluation of the P\'olya density and potentially leads to better mixing with large $N$ but at an increased cost compared to Metropolis-Hastings. 

\subsection{Prior for beta parameters (a,b)}

We also consider a further hierarchical setting where beta parameters $(a,b)$ are no longer fixed but unknown. For simplicity, consider the nonparametric binary regression model where some proper prior $\pi(a,b)$ is placed on $(a,b)$.
Although the inference could be carried out by having an additional step of sampling $(a,b)\mid \lambda$ in Algorithm~\ref{alg:lb-bernoulli} using Metropolis-Hastings, it involves the evaluation of P\'olya densities which can cause numerical issues and can lead to slow mixing of the overall algorithm. 
 
We describe how inference can be still carried out without evaluation of P\'olya densities and jeopardizing the overall mixing behavior. This is done by expanding the block sampling of $(\lambda,\bm\eta)$ in Algorithm~\ref{alg:lb-bernoulli} to the block sampling of $(a,b,\lambda,\bm\eta)$, by sequential updates from $\pi(a,b,\lambda \mid \bm\omega, \bfz)$, and $\pi(\bm\eta \mid a,b,\lambda, \bm\omega, \bfz)$. To update $(a,b,\lambda)$ jointly, one can use particle marginal Metropolis-Hastings algorithm \citep{Andrieu2010-tm}. With $N$ number of particles, it proceeds as: (A) draw a new candidate $(a^\star,b^\star)\sim \tilde{q}(a^\star,b^\star\mid a,b)$ with some proposal distribution $\tilde{q}$, (B) draw new particles $\lambda^{\star(1)}, \dots,\lambda^{\star(N)}\iidsim \mathrm{Polya}(a^\star,b^\star)$, (C) Draw a candidate $\lambda^\star$ among 
$\lambda^{\star(1)}, \dots,\lambda^{\star(N)}$ with probability proportional to $\calL(\lambda^{\star(\ell)})$, $\ell = 1,\dots,N$, and (D) Accept $(a^\star,b^\star,\lambda^{\star})$ and set of particles $\{\lambda^{\star(\ell)}\}_{\ell=1}^N$ with probability 
\[
\min \left\{1, \frac{\pi(a^\star,b^\star)\sum_{\ell = 1}^N\calL(\lambda^{\star(\ell)})}{\pi(a,b)\sum_{\ell = 1}^N \calL(\lambda^{(\ell)})}\times \frac{\tilde{q}(a,b\mid a^\star,b^\star)}{\tilde{q}(a^\star,b^\star\mid a,b)}\right\},
\]
otherwise keep $(a,b,\lambda)$ and current set of particles $\{\lambda^{(\ell)}\}_{\ell=1}^N$. 
The joint update of $(a,b,\lambda)$ from collapsed conditional $\pi(a,b,\lambda \mid \bm\omega, \bfz)$, where $\bm\eta$ is marginalized out, is made possible due to normal-normal conjugacy. Updating $\bm\eta$ is the same as Step 3 of Algorithm~\ref{alg:lb-bernoulli}. 

\color{black}
\subsection{Challenges for developing rejection samplers involving P\'olya densities}
\label{appendix:b3}

Motivated by the broad success of P\'olya-Gamma data augmentation algorithms in posterior computation for logistic-type models \citep{Polson2013-gb}, we discuss similarities and differences between P\'olya and P\'olya-Gamma random variables, particularly in the context of rejection samplers.
The P\'olya-Gamma random variable, denoted as $\omega\sim \textsc{pg}(v,0)$, is defined as an infinite convolution of gamma random variables $ (2\pi^2)^{-1}\sum_{k=0}^\infty g_k/(k+1/2)^2, g_k\iidsim \mathrm{Ga}(v,1)$, and its two-parameter version $\textsc{pg}(v,c)$ arises from an exponential tilting of the $\textsc{pg}(v,0)$ density $\pipg(\omega; v,c)\propto \pipg(\omega; v,0)\exp(-c^2\omega/2)$. 
Compared to P\'olya random variables, there is an intersection between P\'olya-Gamma where $\mathrm{Polya}(a=1/2,b=1/2)$ is equal in distribution to $\textsc{pg}(v=1, c=0)$ scaled by $4\pi^2$, but they generally don't have overlap for other parameter choices.

The P\'olya-Gamma data augmentation scheme facilitates conditionally conjugate updating of the logistic-beta latent parameter $\eta$ in Algorithm~\ref{alg:lb-bernoulli} and \ref{alg:lb-ddp}, and the development of an efficient rejection sampler for $\textsc{pg}(v,c)$ played an important role in its success in many statistical models. 
Similarly, the rejection sampler for Kolmogorov--Smirnov outlined in \citet[][\S 5.6]{Devroye1986-gj} has been utilized in the data augmentation scheme of \citet{Holmes2006-np}. 
Building rejection samplers for $\textsc{pg}(1,c)$ and Kolmogorov--Smirnov crucially rely on two different alternating sum representations of Jacobi theta functions \citep{Devroye2009-fn}, which are essential to meet the monotonicity condition of the rejection sampler based on alternating series \citep{Devroye1986-gj}. To the best of our knowledge, we are only aware of such dual representations for two highly special cases of $\mathrm{Polya}(1/2,1/2)$ and $\mathrm{Polya}(1,1)$ that correspond to P\'olya-Gamma and Kolmogorov--Smirnov random variables, respectively, but not for general choices of $(a,b)$ of P\'olya distribution; see Table 1 of \citet{Biane2001-nz} and Theorems 6.1 and 6.3 of \citet{Salminen2024-sr}. Another major hurdle in developing a rejection sampler is that the multivariate normal density in expression \eqref{eq:polyaupdate} does not have a simple form as a function of $\lambda$, in contrast to a simple exponential tilting in the case of \citet{Polson2013-gb}. Solving the above problems is an interesting direction for future research.

\section{Details of simulation analysis and application study}
\label{appendix:c}
\subsection{Nonparametric binary regression simulation studies}
\label{appendix:c1}
We first provide a detailed description of the quantities we used to evaluate the simulation results. The effective sample size is calculated with an autoregressive approach implemented in R package \texttt{coda} \citep{Plummer2006-ii}. Similarly, the multivariate effective sample size is calculated with R package \texttt{mcmcse} \citep{Flegal2017-jl}. The root mean squared error and mean absolute error are calculated based on a posterior mean estimate. The mean continuous ranked probability score is calculated by $n^{-1}\sum_{i=1}^n S(\hat{F}_{\theta_i}, \theta_i^*)$, where $S(F,\theta)$ is a continuous ranked probability score (lower the better) and $\hat{F}_{\theta_i}$ is an empirical distribution of $\theta_i$ obtained from posterior samples of $\theta_i$. 
All algorithms are executed on an Intel(R) Xeon(R) Gold 6132 CPU with a 96GB memory environment. The algorithm running times are defined as the difference between wall-clock time (in seconds) from the start to the finish of Markov chain Monte Carlo algorithm, comprising a total of 2,000 iterations. This excludes pre-computation time, such as the time required to compile \texttt{Stan} programs.

The range parameter of Mat\'ern kernel $\varrho$ is assumed to be unknown, and those are learned from data when fitting both latent logistic-beta process and Gaussian copula models. For the latent logistic-beta process model, we introduce a discrete uniform prior distribution on the set $\{0.01,0.02,\dots,0.5\}$ (with probability 0.02 each), and for the Gaussian copula model, we use continuous uniform prior $\varrho\sim \mathrm{Unif}(0.01, 0.5)$ as \texttt{Stan} can only handle continuous priors. The smoothness parameter is fixed as $\nu = 1.5$ for both models and the distance calculations $\|x_i - x_j\|$ in Mat\'ern kernels are based on Euclidean distance. 

\begin{table}[h]
\caption{Nonparametric binary regression simulation results comparing computational benefits of different strategies. Entries shown are average over 100 replicates, with Monte Carlo standard error shown in parentheses. ESS, effective sample size of $\lambda$; ESS/sec, ESS divided by running time in seconds; Acc. rate, acceptance rate of $\lambda$.
}
\footnotesize
\centering
\begin{tabular}{ccc c c c}
    \toprule
       Data generation & \multicolumn{2}{c}{Algorithm settings} & ESS & ESS/sec & Acc. rate (\%)\\
    \midrule
    \multirow{6}{*}{\makecell{LBP, \\ $\varrho = 0.1$}} & \multirow{2}{*}{\makecell{Non-blocked,\\ Indep. M-H}} & Adapted & 7.89 (0.35) & 0.14 (0.01)& 13.60 (0.26) \\
      & & Non-adapted & 7.13 (0.37) & 0.13 (0.01) & 12.44 (0.31) \\
    \cmidrule(lr){2-6}
      & \multirow{2}{*}{\makecell{Blocked,\\Indep. M-H}} & Adapted & 245.08 (12.86) & 3.35 (0.18) & 54.28 (1.03) \\
     & & Non-adapted & 196.35 (11.28)  & 2.69 (0.15) & 49.39 (1.39) \\
      \cmidrule(lr){2-6}
     &\multirow{2}{*}{\makecell{Blocked,\\Particle Gibbs}} & Adapted & 358.10 (18.57) & 2.33 (0.12)& 84.50 (0.42) \\
     & & Non-adapted & 352.73 (18.83) & 2.34 (0.13) & 81.35 (1.02) \\
     \midrule
    \multirow{6}{*}{\makecell{LBP, \\ $\varrho = 0.2$}} & \multirow{2}{*}{\makecell{Non-blocked,\\ Indep. M-H}} & Adapted & 7.31 (0.32) & 0.11 (0.00) & 13.42 (0.22) \\
     & & Non-adapted & 6.58 (0.32) & 0.10 (0.00) & 12.48 (0.31) \\
      \cmidrule(lr){2-6}
     & \multirow{2}{*}{\makecell{Blocked,\\Indep. M-H}}  & Adapted & 257.01 (16.32) & 2.89 (0.18) &  62.26 (1.12) \\
     & & Non-Adapted &  247.83 (16.53) &  2.82 (0.19) & 57.60 (1.47)\\
      \cmidrule(lr){2-6}
     &\multirow{2}{*}{\makecell{Blocked,\\Particle Gibbs}}  & Adapted & 360.48 (21.29) & 1.96 (0.12) & 86.11 (0.40) \\
     & & Non-adapted & 366.55 (24.75) & 1.99 (0.13) & 84.08 (0.64) \\
     \midrule
    \multirow{6}{*}{\makecell{LBP, \\ $\varrho = 0.4$}} & \multirow{2}{*}{\makecell{Non-blocked,\\ Indep. M-H}} & Adapted & 6.40 (0.29) & 0.09 (0.00) & 13.01 (0.22) \\
     & & Non-adapted & 6.56 (0.32) & 0.10 (0.00) & 12.88 (0.31) \\
     \cmidrule(lr){2-6}
     & \multirow{2}{*}{\makecell{Blocked,\\Indep. M-H}} & Adapted & 368.26 (17.54) & 4.12 (0.20) & 66.12 (0.98) \\
     & & Non-adapted & 328.72 (19.00) & 3.67 (0.21) & 61.45 (1.32) \\
      \cmidrule(lr){2-6}
     &\multirow{2}{*}{\makecell{Blocked,\\Particle Gibbs}} & Adapted & 468.76 (25.98) & 2.52 (0.14) & 87.09 (0.29) \\
     & & Non-adapted & 445.09 (25.04) & 2.40 (0.14) & 85.18 (0.58) \\
    \bottomrule
  \end{tabular}

  \label{table:sim1-add}
\end{table}

We first present additional simulation results in Table~\ref{table:sim1-add}, similar to those in Table~\ref{table:sim1}, to analyze and compare the computational strategies involved. 
First, we compare with a naive version of the Gibbs sampler where Step 2 of Algorithm \ref{alg:lb-bernoulli} is replaced with sampling $\lambda$ from full conditional $\pi(\lambda\mid \bm\omega, \bm\eta, \bfz )$, referred to as ``Non-blocked''. 
Also, instead of independent Metropolis-Hastings, we employ a particle Gibbs sampler for sampling $\lambda$ with $N=10$, combined with a blocked Gibbs sampler. 

Table~\ref{table:sim1-add} reports the mixing behavior and sampling efficiency of $\lambda$ of algorithms under different computational strategies. Here, the acceptance rate for the particle Gibbs sampler is defined as the number of iterations when one of the new $N$ candidates is accepted as a new sample, divided by the total number of iterations. 
The blocking scheme has a substantial impact on the mixing behavior, leading to an effective sample size increase of up to 50 times. 
Also, it shows that employing a particle Gibbs sampler improves the effective sample size of $\lambda$, but at the cost of decreased sampling efficiency measured by effective sample size divided by seconds. Utilizing the adaptive P\'olya proposal with a particle Gibbs sampler shows less significant impact compared to the adaptive P\'olya proposal with independent Metropolis-Hastings, as both aim to improve the mixing behavior of $\lambda$. These results highlight the benefits of the blocking strategy, which is made possible due to the normal mixture representation of logistic-beta. It also suggests that the particle Gibbs sampler could enhance the mixing behavior, which is particularly beneficial when the acceptance rate from independent Metropolis-Hastings is low. While this approach may improve mixing, it does not guarantee enhanced sampling efficiency.

\begin{table}
\caption{Nonparametric binary regression simulation results comparing latent logistic-beta process (LBP) and Gaussian copula models (GauCop). 
  Entries are averaged over 100 replicates, with Monte Carlo standard errors in parentheses. RMSE, root mean squared error; CRPS, continuous ranked probability score; mESS/sec, multivariate effective sample size of latent success probabilities divided by running time in seconds.}
  \footnotesize
  \centering
  \begin{tabular}{ccc cc cc}
    \toprule
   \multirow{2}{*}{\makecell{Data \\ generation}} & \multirow{2}{*}{Model} & \multicolumn{2}{c}{RMSE $\times 100$} & \multicolumn{2}{c}{mean CRPS $\times 100$} & \multirow{2}{*}{mESS/sec}  \\
   \cmidrule(lr){3-4} \cmidrule(lr){5-6}
   &  & training & test & training & test & \\
    \midrule
    \multirow{2}{*}{\makecell{LBP, \\ $\varrho = 0.1$}} & LBP & 11.59 (0.15) & 12.17 (0.20) & 6.31 (0.10) & 6.62 (0.12) & 21.27 (0.24)\\
     &GauCop & 11.66 (0.15) & 12.19 (0.20) & 6.35 (0.09) &6.65 (0.12) & 0.48 (0.01)\\
    \midrule
    \multirow{2}{*}{\makecell{LBP, \\ $\varrho = 0.2$}} & LBP &8.54 (0.18) & 8.73 (0.19) & 4.67 (0.11) & 4.77 (0.11) & 17.87 (0.16) \\
     &GauCop &  8.59 (0.17) &  8.76 (0.18) & 4.70 (0.10) & 4.79 (0.11) & 0.44 (0.01)\\
  \midrule
    \multirow{2}{*}{\makecell{LBP, \\ $\varrho = 0.4$}} & LBP & 6.12 (0.19) & 6.16 (0.19) & 3.41 (0.11) & 3.43 (0.11) & 17.48 (0.16) \\
     &GauCop &  6.15 (0.17) & 6.19 (0.18) & 3.43 (0.10) & 3.44 (0.11)  & 0.47 (0.01)\\
    \bottomrule
  \end{tabular}
  
  \label{table:sim2-lbp}
\end{table}

Table~\ref{table:sim2-lbp} provides simulation results that compare latent logistic-beta process and Gaussian copula models that are similar to Table~\ref{table:sim2} but with data generation based on the latent logistic-beta process model. Compared to Table~\ref{table:sim2}, where the Gaussian copula model achieved slightly better predictive performance, Table~\ref{table:sim2-lbp} shows that the latent logistic-beta process model can achieve slightly better predictive performance of of latent success probabilities when the model is correctly specified. In terms of sampling efficiency, the latent logistic-beta process clearly outperforms the Gaussian copula model.

\begin{table}
\caption{Additional simulation results comparing latent logistic-beta process (LBP) and Gaussian copula models (GauCop). Entries shown are average over 100 replicates, with Monte Carlo standard error in parentheses. MAE, mean absolute error of inferred \& predicted success probabilities on training \& test data; mean ESS/sec,
Univariate effective sample sizes of latent success probabilities divided by running time in seconds, averaged across training dataset.}
  \footnotesize
  \centering
  \begin{tabular}{ccc ccc}
    \toprule
   \multirow{2}{*}{\makecell{Data\\generation}} & \multirow{2}{*}{Model} & \multicolumn{2}{c}{MAE $\times 100$} & \multirow{2}{*}{mean ESS/sec} & \multirow{2}{*}{time (sec)}  \\
   \cmidrule(lr){3-4} 
   &  & training & test & & \\
    \midrule
    \multirow{2}{*}{\makecell{GauCop, \\ $\varrho = 0.1$}} & LBP & 9.34 (0.13) & 9.67 (0.15) & 9.44 (0.16) & 73.44 (0.48)\\
     &GauCop & 9.24 (0.12) & 9.60 (0.14) & 0.49 (0.01) & 2135.52 (44.88)\\
    \midrule
    \multirow{2}{*}{\makecell{GauCop, \\ $\varrho = 0.2$}} & LBP &6.79 (0.14) & 6.90 (0.14) & 7.36 (0.15) & 89.86 (0.38) \\
     &GauCop &  6.74 (0.14) &  6.86 (0.15) & 0.41 (0.01) & 2882.76 (255.82)\\
    \midrule
    \multirow{2}{*}{\makecell{GauCop, \\ $\varrho = 0.4$}} & LBP & 4.82 (0.14) & 4.84 (0.14) & 7.45 (0.21) & 89.98 (0.34) \\
     &GauCop &  4.80 (0.14) & 4.83 (0.14) & 0.47 (0.01) & 2320.49 (84.65)\\
    \midrule
    \multirow{2}{*}{\makecell{LBP, \\ $\varrho = 0.1$}} & LBP & 9.03 (0.14) & 9.54 (0.17) & 9.31 (0.20) & 73.50 (0.46)\\
     &GauCop & 9.10 (0.13) & 9.58 (0.17) & 0.48 (0.01) & 2182.61 (47.37)\\
    \midrule
    \multirow{2}{*}{\makecell{LBP, \\ $\varrho = 0.2$}} & LBP & 6.64 (0.15) & 6.78 (0.16) & 7.36 (0.18) & 88.65 (0.48) \\
     &GauCop &  6.67 (0.15) &  6.81 (0.15) & 0.45 (0.01) & 2356.76 (50.61)\\
    \midrule
    \multirow{2}{*}{\makecell{LBP, \\ $\varrho = 0.4$}} & LBP & 4.84 (0.15) & 4.86 (0.16) & 7.53 (0.21) & 89.91 (0.46) \\
     &GauCop &  4.85 (0.14) & 4.88 (0.15) & 0.48 (0.01) & 2305.82 (85.17)\\
    \bottomrule
  \end{tabular}  
  \label{table:sim2-add}
\end{table}

We provide additional summaries in Table~\ref{table:sim2-add} that compare latent logistic-beta process and Gaussian copula models. Similar to Table~\ref{table:sim2}, Table~\ref{table:sim2-add} shows that the in-sample and predictive performance of the latent logistic-beta process model in terms of mean absolute error is comparable to the Gaussian copula model, showing the flexibility of logistic-beta process despite model misspecification. When data are generated from the latent logistic-beta process model, the in-sample and predictive performance of the latent logistic-beta process model is slightly better than the Gaussian copula model. 
We also present sampling efficiency in terms of the average of univariate effective sample size and algorithm running times. Logistic-beta process models have a clear computational advantage, likely due to the 
conditionally conjugate updating scheme with an efficient blocked Gibbs sampling algorithm.

\subsection{Bayesian density regression simulation studies}

First, we provide a detailed description of the regression errors and density errors we used to evaluate the simulation results, where we closely followed \citet{Wade2025-ef}. We consider $n_{\mathrm{test}} = 100$ test covariates $x^*_1,\dots,x^*_{100}$ equally spaced over the domain, $[0,1]$ for scenario A and $[-2,10]$ for scenario B. Also, we consider the grid of responses $\{y_g\}$ of size $500$, equally spaced over $[-1,2]$ for scenario A and $[-1,10]$ for scenario B. The regression errors are calculated as
\[
\text{Regression error} = \left[\frac{1}{n_{\mathrm{test}}}\sum_{i=1}^{n_\mathrm{test}}\{m(y\mid x_i^*)- \hat{m}(y\mid x_i^*)\}^2\right]^{1/2}
\]
where we used the posterior mean estimate for $\hat{m}(y\mid x_i^*)$. Also, the density errors are calculated as 
\[
\text{Density error} =  \frac{1}{n_{\mathrm{test}}}\sum_{i=1}^{n_\mathrm{test}}\sum_{g=1}^{500}|p(y_g\mid x_i^*) - \hat{p}(y_g\mid x_i^*)|(y_{2}-y_1)
\]
where we used the posterior mean estimate for $\hat{p}(y_g\mid x_i^*)$.

We set the upper bound on the number of
components as $H=20$ for all models. 
For the priors on atom processes, we used normal prior on regression coefficients $(\beta_{0h},\beta_{1h})^\T\iidsim N_2(\bm{0}_2, 10^2\bfI_2)$ and gamma prior on precision $\tau_h\iidsim \mathrm{Ga}(1, 1)$ for $h=1,\dots,H$. 
All algorithms are executed on Intel(R) Xeon(R) Gold 6336Y 2.40GHz CPU. For all three models, we run a total of 25,000 Markov chain Monte Carlo iterations, and the first 20,000 samples are discarded as burn-in.

\begin{figure}
    \centering
    \includegraphics[width=\linewidth]{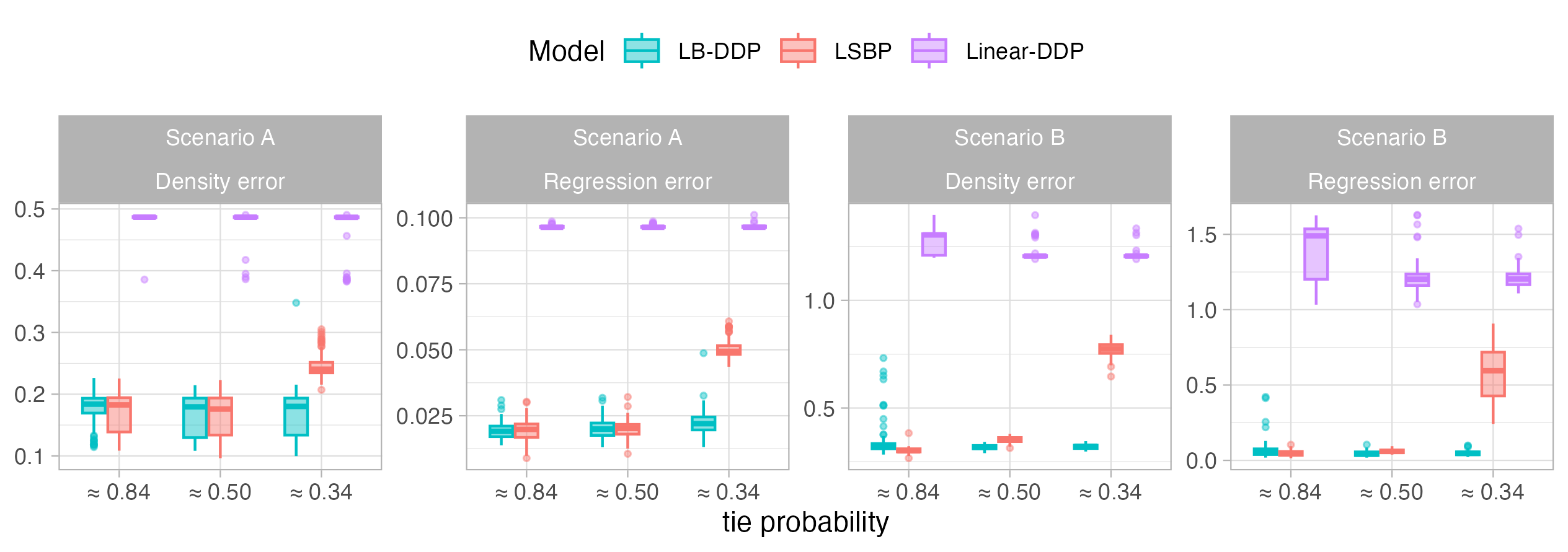}
    \caption{Density regression simulation results under two different data generation scenarios with sample size $n=2000$ and three different hyperparameter settings based on prior tie probability. Density error and regression error quantify the accuracy of conditional density and conditional mean estimates averaged across test covariates, respectively.}
    \label{fig:sim2boxplotn2000}
\end{figure}

Figure~\ref{fig:sim2boxplotn2000} provides density regression simulation results that are similar to Figure~\ref{fig:sim2boxplotn500} but with sample size $n=2000$. Compared to Figure~\ref{fig:sim2boxplotn500}, the performance of the linear DDP mixture model is not improved in both scenarios despite the increased sample size. This highlights the benefits of having covariate-dependent mixture weights, where the linear DDP mixture model is not flexible enough to accommodate the complexities. The LB-DDP and LSBP mixture models show improved performance, but similar to Figure~\ref{fig:sim2boxplotn500}, the LSBP mixture model is also sensitive to the choice of hyperparameter for the setting when the prior tie probability is around 0.34.

\subsection{Pregnancy outcome data analysis}
\label{appendix:c2}

We divide the collaborative perinatal project data available from the \texttt{BNPmix} package into two groups based on maternal smoking status during pregnancy, resulting in data sizes of $n=1023$ (smoking) and $n=1290$ (nonsmoking), respectively. 
We centred and scaled the data before the analysis, and we set the upper bound on the number of components as $H=20$ for both the logistic-beta dependent Dirichlet process and the logit stick-breaking process. We considered semiconjugate priors for mixture component parameters in expression \eqref{eq:lbddpmixture2} as $(\beta_{0h},\beta_{1h})^\T \iidsim N_2(\bm{0}_2,\bfI_2),\tau_h\iidsim \mathrm{Ga}(1,1)$, $h=1,\dots,H$ for both models.

We run Algorithm~\ref{alg:lb-ddp} for logistic-beta dependent Dirichlet process, and Algorithm 1 of \cite{Rigon2021-ir} for logit stick-breaking process to carry out the posterior inference, both implemented with \texttt{R} for time comparison purposes. 
We ran a total of 35,000 Markov chain Monte Carlo iterations for both models, and the first 5,000 samples were discarded as burn-in, executed on an Intel(R) Xeon(R) Gold 6132 CPU with a 96GB memory environment. 
The analysis of trace plots shows satisfactory convergence for both models. 
The running time of Markov chain Monte Carlo samplers under settings 1, 2, and 3 are approximately 5.4, 7.5, and 10.2 minutes for the logistic-beta dependent Dirichlet process and 4.9, 5.6, and 6.2 minutes for the logit stick-breaking process, respectively, for the analysis of smoking group data. This shows the computational tractability of the logistic-beta dependent Dirichlet process model is similar to the logit stick-breaking process model, and we note that substantial time improvements are possible with lower-level programming languages.

\begin{figure}
\centering
    \includegraphics[width=\textwidth]{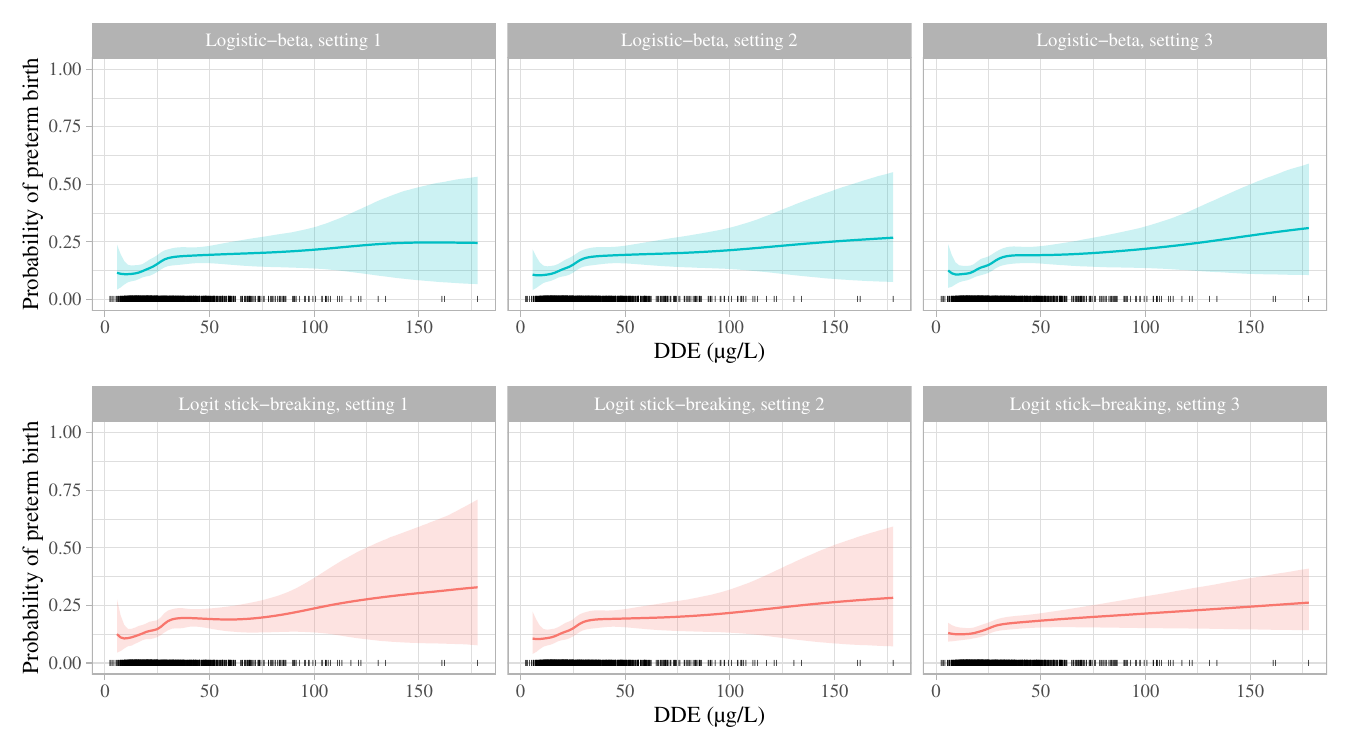}
    \caption{Estimated probability of preterm birth with 95\% credible intervals against DDE exposure level from the logistic-beta dependent Dirichlet process model (top) and the logit stick-breaking process model (bottom) with three different hyperparameter settings corresponding to prior tie probabilities approximately 0.84, 0.50, and 0.34. DDE exposure data are represented by small vertical bars.}
    \label{fig:ddp_pretermbirth_nonsmoke}
\end{figure}

Figure~\ref{fig:ddp_pretermbirth_nonsmoke} presents the estimated preterm birth probabilities given DDE exposures for the non-smoking group of size $n=1290$. Similar to Figure~\ref{fig:ddp_pretermbirth_smoke}, the estimated preterm birth probabilities show an overall increasing pattern as the DDE exposure level rises. However, unlike the logistic-beta dependent Dirichlet process, the hyperparameter choice of the logit stick-breaking process model has a substantial impact on uncertainty estimates. 

\end{document}